\documentclass[aps,nofootinbib,twocolumn,pra]{revtex4-2}

\usepackage[utf8]{inputenc}
\usepackage[margin=1in]{geometry}
\usepackage{fancyhdr}
\usepackage[usenames, dvipsnames]{color}

\usepackage{mathpazo}
\usepackage{amssymb}
\usepackage{amsmath}
\usepackage{mathtools}
\usepackage{relsize}
\usepackage{hyperref}
\usepackage{amsthm}
\usepackage{physics}
\usepackage{mathrsfs}

\usepackage{matlab-prettifier}

\usepackage{float} 
\allowdisplaybreaks 
\usepackage{hyperref} 
\usepackage{mwe}
\usepackage{listing}
\usepackage[en-US]{datetime2}
\usepackage{placeins}
\usepackage{capt-of}
\usepackage{pgf,tikz}
\usetikzlibrary{arrows}
\usetikzlibrary{shapes}
\usetikzlibrary{patterns}
\usepackage{enumitem}
\usepackage{subcaption}
\usepackage{comment}

\newcommand{\wt}{\widetilde}

\newcommand{\ve}{\varepsilon}
\newcommand{\mrm}{\mathrm}
\newcommand{\mbb}{\mathbb}
\renewcommand{\ol}{\overline}

\newcommand{\opvec}{\operatorname{vec}}

\renewcommand{\rank}{\operatorname{rank}}

\newcommand{\SR}{\mathrm{SR}}

\newcommand{\cC}{\mathcal{C}}

\newcommand{\cE}{\mathcal{E}}
\newcommand{\cF}{\mathcal{F}}

\newcommand{\cI}{\mathcal{I}}

\newcommand{\cO}{\mathcal{O}}
\newcommand{\cP}{\mathcal{P}}

\newcommand{\Lin}{\mathrm{L}}
\newcommand{\Trans}{\mathrm{T}}
\newcommand{\Pos}{\mathrm{Pos}}
\newcommand{\Herm}{\mathrm{Herm}}
\newcommand{\Channel}{\mathrm{C}}

\newcommand{\Density}{\mathrm{D}}

\newenvironment{mylist}[1]{\begin{list}{}{
	\setlength{\leftmargin}{#1}
	\setlength{\rightmargin}{0mm}
	\setlength{\labelsep}{2mm}
	\setlength{\labelwidth}{8mm}
	\setlength{\itemsep}{0mm}}}
	{\end{list}}

\newcounter{problemcounter}

\theoremstyle{definition}
\newtheorem{theorem}{Theorem}
\newtheorem{lemma}[theorem]{Lemma}

\newtheorem{definition}{Definition}

\newtheorem{corollary}{Corollary}

\newtheorem{example}{Example}
\newtheorem{remark}{Remark}
\newtheorem{proposition}{Proposition}

\definecolor{cool_green}{rgb}{0.0, 0.5, 0.0}

\newcommand{\edit}[1]{{\color{black}#1}}

\begin{document}
\title{Revisiting Pure State Transformations with Zero Communication}
\author{Ian George}
\affiliation{Electrical and Computer Engineering Department, University of Illinois at Urbana-Champaign}
\author{Eric Chitambar}
\affiliation{Electrical and Computer Engineering Department, University of Illinois at Urbana-Champaign}
\date{\today}
\begin{abstract}
It is known that general convertibility of bipartite entangled states is not possible to arbitrary error without some classical communication.  While some trade-offs between communication cost and conversion error have been proven, these bounds can be very loose.  In particular, there are many cases in which tolerable error might be achievable using zero-communication protocols.  In this work we address these cases by deriving the optimal fidelity of pure state conversions under local unitaries as well as local operations and shared randomness (LOSR).  We also use these results to explore catalytic conversions between pure states using zero communication.
\end{abstract}
\maketitle

\section{Introduction}
The theory of quantum mechanics through the lens of information and vice versa \cite{Wheeler-1989a, Zurek-1990a, Landauer-1991a} has afforded the physicist and the information scientist alike with a new way to view the objects and long-term goals of their study. No better example of this can be found than quantum resource theories. Quantum resource theories specify the relevant physical property in such a manner as to better tease apart the complexities of quantum mechanics while also establishing what tasks may be achieved with said resource \cite{Chitambar-2019a}. Perhaps the earliest example of such a resource theory is the resource theory of entanglement. Entanglement may be viewed as a form of correlation that does not exist in the classical world \cite{Bell-1964}. Roughly speaking, the resource theory of entanglement asks (1) what tasks may be performed better using entangled states and (2) how entangled states may be converted from one to another under some class of free operations. 

The most standard view of the resource theory of entanglement considers the set of free operations to be local operations and classical communication (LOCC) which captures the `distant lab' paradigm where two (or more) parties share an entangled state in spatially separated labs and they can only perform operations on their respective portions and exchange classical information (See Fig. \ref{fig:intro-diagrams}). Not only is this the most standard set of free operations, but in some respect it seems minimal. Indeed, Hayden and Winter showed that to convert one (pure) entangled state to another to sufficiently small precision requires a certain amount of communication between labs, regardless of how many auxiliary EPR pairs they share \cite{Hayden-2003a} (see also \cite{Harrow-2004a}). This is distinct not only from the classical setting \cite{Wyner-1975a}, but also from quantum states that are not entangled \cite{Hayashi-2006a, George-2022c}. However, the results of Hayden and Winter, while fundamental, do not give us a complete picture of the tradeoff between communication and achievable tolerated error in pure state conversions. Indeed, it is easy to find examples of state conversions which, according to the best known lower bounds, still may be possible to perform with a tolerated error of $1\%$ using no communication (see Example \ref{ex:necessity-of-comm} of Section \ref{sec:background}).  This show\edit{s} that a relatively large gap in our understanding of zero-communication entanglement transformations still persists, and one we aim to address in this work.

\begin{figure}
\centering
\begin{subfigure}[b]{0.46\columnwidth}
    \begin{center}
    \includegraphics[width=\columnwidth]{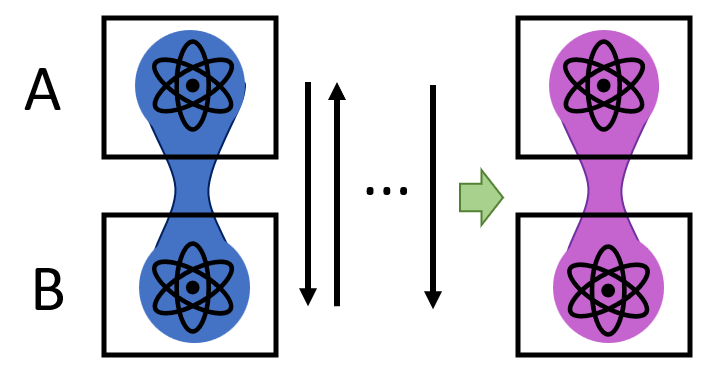}
    \end{center}
   \caption{}
\end{subfigure}
\begin{subfigure}[b]{0.46\columnwidth}
   \begin{center}
    \includegraphics[width=\columnwidth]{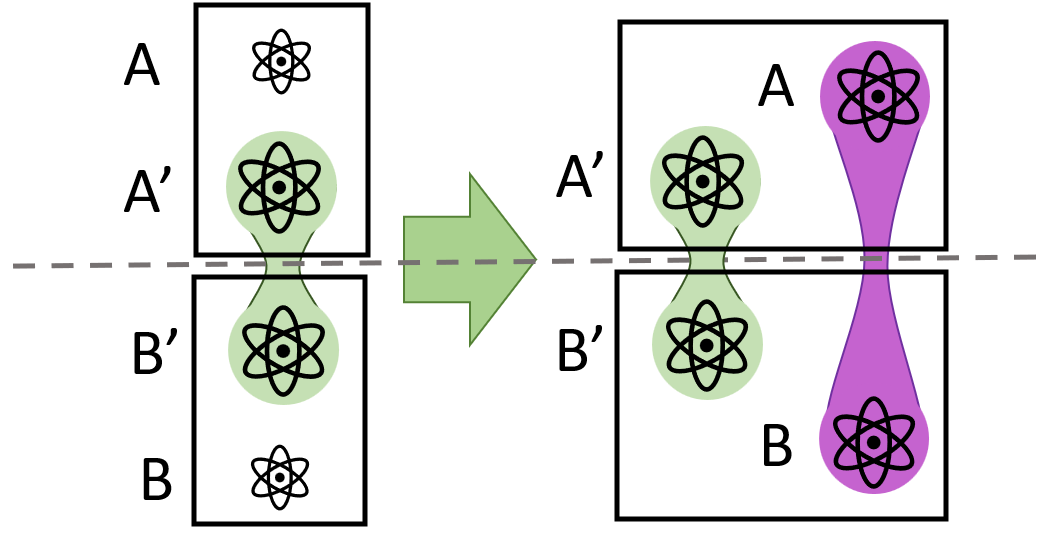}
    \end{center}
   \caption{}
\end{subfigure}
\caption{\small{Conversion of pure states in distant labs. 
(a) The LOCC model where communication is exchanged. (b) The embezzling of quantum states where an auxiliary entangled state is used. This may be seen as a special case of catalytic conversion.}}
\label{fig:intro-diagrams}
\end{figure}

Moreover, the tools we develop to address this problem will also allow us to study pure state transformations using shared auxiliary entanglement. The operational paradigm in which parties are allowed to use arbitrary pre-shared entanglement but no communication is known as local operations and shared entanglement (LOSE) \cite{Schmid-2021a}. \edit{Without restriction on the entangled resource, pure state convertibility $\ket{\psi}_{AB}\to\ket{\phi}_{AB}$ under LOSE is trivial since Alice and Bob could always just demand $\ket{\phi}$ as their pre-shared entanglement and then throw away $\ket{\psi}$ when it is given.  However, if one tries to minimize the amount of pre-shared entanglement (under some measure), the problem becomes quite interesting. A more demanding setting is to minimize the amount of pre-shared entanglement $\ket{\omega}$ while requiring }it is also returned to the target state $\ket{\phi}$, i.e. $\ket{\psi}_{AB}\otimes\ket{\omega}_{A'B'}\to\ket{\phi}_{AB}\otimes\ket{\omega}_{A'B'}$ for auxiliary pre-shared entanglement $\ket{\omega}_{A'B'}$.  Transformations of this form are known as catalytic transformations with $\ket{\omega}_{A'B}$ being the catalyst \edit{(See \cite{Datta-2023a,Lipka-Bartosik-2023a} for reviews on catalytic transformations)}.

\edit{When instead the catalyst only needs to be returned approximately, we distinguish it from the zero error case by calling it an `embezzler' or `embezzling state.'} Remarkably, \edit{van Dam and Hayden showed that there exists a family of `universal embezzling states' under local operations \cite{van-2003a}. That is, they found a family of states $\{\ket{\mu(n)}_{A'B'}\}_{n \in \mbb{N}}$ such that for any tolerated non-zero error $\varepsilon$, one can always prepare a pure state $\ket{\psi}_{AB}$ using a member of this family and zero communication. Perhaps even more surprising, the authors showed that this family of embezzlers is `optimal' in a certain sense. In particular, they showed that if one strengthens the class of transformations to local operations and classical communication (LOCC) and replaces the appropriate universal embezzler state $\ket{\mu(n_{\ve})}$ for some $n_{\ve} \in \mbb{N}$ to a state-dependent embezzler state $\ket{\zeta(\psi,\ve)}$ that also depends on the tolerated error $\ve$, the dimension of the entanglement in $\ket{\mu(n_{\ve})}$ and $\ket{\zeta(\psi,\ve)}$ will become nearly the same as $\ve$ approaches zero.}  This near optimality along with Hayden and Winter's result has, understandably, largely ceased the study of entanglement transformations with zero communication, because when one needs entanglement transformations without communication, one uses embezzlement \cite{Bennett-2014a, Anshu-2021a}.\footnote{The notable exceptions to this halted topic of research has been the consideration of special embezzling families \cite{Leung-2014a} and the correlated sampling lemma \cite{Dinur-2015a}, which may be viewed as a variation of embezzling.} It is however not clear what is the necessary error for embezzlement to become near optimal, which could be relevant in practical settings. Indeed, for any tolerated error, it is easy to find sufficient conditions on pure states to be converted with no catalyst at all (Example \ref{ex:on-necc-of-emb} of Section \ref{sec:background}). This is an indication that we also do not understand embezzling sufficiently well.

\subsection{Summary of Results}
\begin{figure}
    \centering
    \includegraphics[width=\columnwidth]{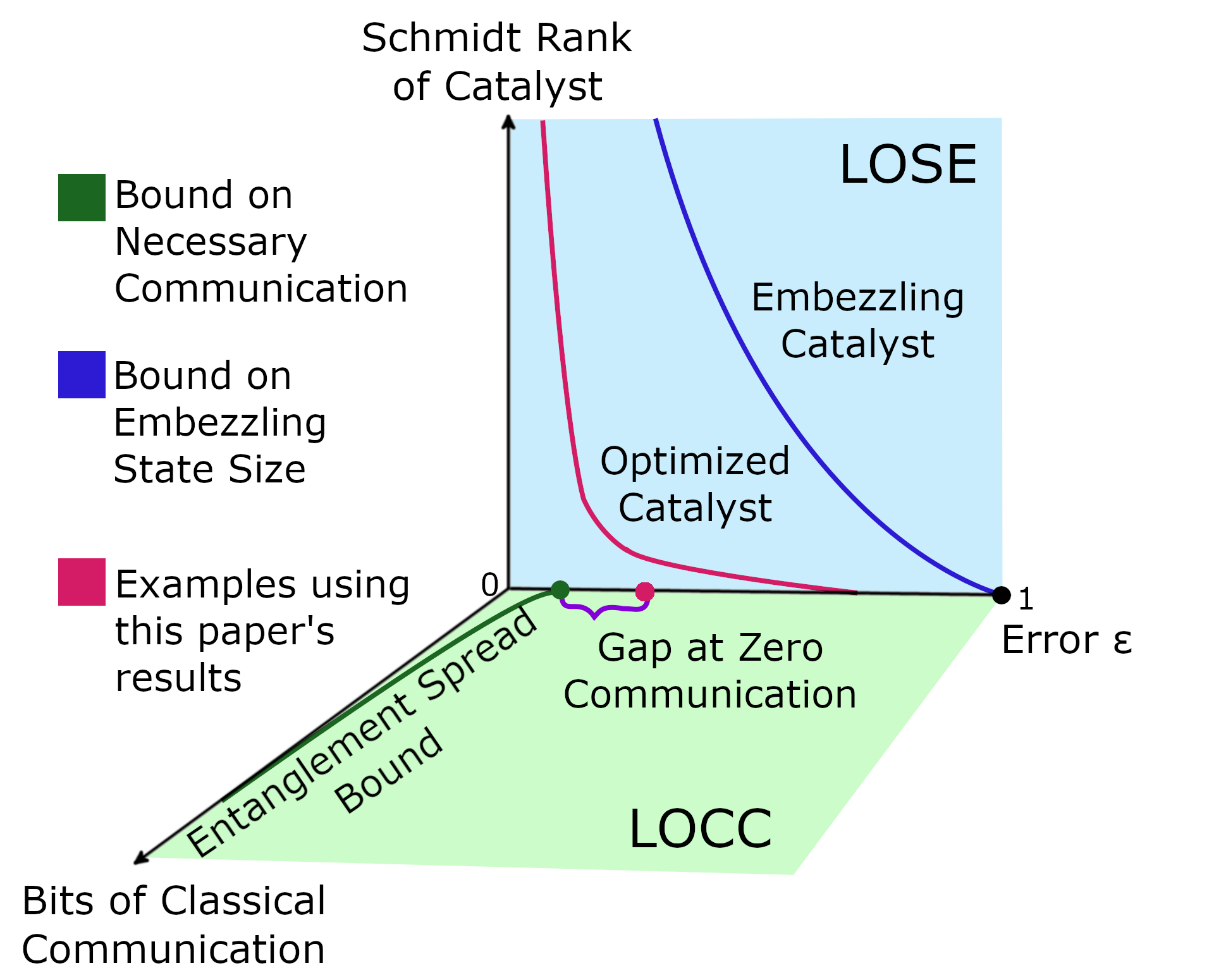}
    \caption{\small Comparison of \cite{van-2003a} (blue),\cite{Hayden-2004a} (green), and this work's results (pink). \cite{Hayden-2004a} finds lower bounds on the classical communication necessary to convert one state to another, but in the zero communication setting the bound is loose \edit{(see Examples \ref{ex:necessity-of-comm} and \ref{ex:tightened-error} for an example of the depicted gap)}. We find methods for solving this exactly (Section \ref{sec:single-copy}), which establishes that communication is necessary for larger tolerated errors. \cite{van-2003a} establishes a method for pure state transformations with zero communication with massive amounts of entanglement. This result is known to be almost optimal for sufficiently small error. Nonetheless, we find we find the necessary resources can be too strong for a relevant error range, even if ultimately it is optimal (Section \ref{sec:cat-conv}).} 
    \label{fig:res-region}
\end{figure}
The primary aim of this work is to provide tighter lower bounds on the error in pure state entanglement convertibility with zero communication. A high level comparison of our results to the aforementioned work on this topic are presented in Fig.\ \ref{fig:res-region}.  This depicts a `one-shot resource tradeoff' region that must contain the `true' one-shot resource tradeoff surface for a given pure state conversion. Hayden and Winter's result provides a lower bound on the achievability independent of the amount of shared maximally entangled states, but their result can be too loose when considering zero communication. van Dam and Hayden's result provides an outer bound on the achievability surface on the face pertaining to LOSE, but their result in fact can be too loose when the error is not sufficiently small. In this work, our results allow one to exactly solve the minimal error in the zero communication setting and also provide significantly tighter bounds than quantum embezzling for a relevant region on the LOSE face (See Fig.\ \ref{fig:res-region}).

To formally establish our results, we reduce the class of questions regarding optimal pure state conversion to optimization problems that only concern non-negative vectors. This is because of a bijection between the equivalence classes of pure states under local unitaries--- which are defined solely by their Schmidt coefficients--- and the probability simplex. We do this by showing the optimal fidelity of pure state transformations with local unitaries is efficiently computable. Of course, in general one would not expect local unitaries to be the optimal strategy and we build on this result to present a non-convex optimization program over an optimization variable with bounded dimension. An immediate corollary of this result is the impossibility of pure state conversions with zero communication for negligible error. We also present efficient computable upper bounds on the achievable error using a semidefinite programming (SDP) relaxation. We also show that in the case where either the seed (i.e. initial) or target state is a two-qubit state, the local unitary strategy is optimal. However, we can show for larger dimensions this is not the case.

Having established general properties in the single copy case, we move to the multiple copy case, i.e. where the seed and/or target state is of independent and identically distributed (i.i.d.) form. This is standard in determining the rate of converting one state to another. In particular, we consider dilution and distillation where the seed state or target state respectively is many copies of a maximally entangled state and show these are convex optimization programs and may be seen as involving the Ky-Fan norms when extended to the regime where they are not a norm. Lastly, in a sense extending our earlier two-qubit results, we establish that if the target state is an $n-$fold copy of a two-qubit entangled state and the seed state's Schmidt rank is less than the target state, then local unitaries are the optimal strategy.

Finally, given these results, we turn our attention to quantum embezzlement. We begin by noting that the correspondence between Schmidt coefficients and probability distributions means that quantum embezzlement implies a classical equivalent we call randomness embezzlement. We then proceed to use our new tools to consider the problem of \edit{approximate catalytic} pure state conversion under local unitaries, in effect a generalization of \edit{traditional} embezzling, and compare it to embezzling. We show in particular that at least in general the optimality of the embezzling states is only for very small errors. Indeed, we show for reasonable tolerable errors, the embezzling state may have a Schmidt rank of many orders of magnitude larger than a \edit{state-dependent embezzler}. This may have practical relevance and strongly refines our understanding of pure state transformations under LOSE.

\subsection{Relation to Previous Work}
\edit{There has been a great deal of work on embezzlement within communication theory as well as to some degree within quantum thermodynamics. There has also been work on state conversion under LOSR. As such, we briefly state how our work relates to these and other mathematical methodologies. 

The major technical results of this work focus on the maximum fidelity of conversion between bipartite entangled pure states under local unitaries or local operations and shared randomness. This may be viewed as a variation of optimal conversion distance, which is often described in terms of trace distance (See \cite{Gour-2024a} for a discussion of conversion distance in terms of trace norm for general resource theories.)  This significantly differs from previous works on conversion of bipartite states under LOSR \cite{Buscemi-2012a,Schmid-2021a}, which establish equivalent conditions for there being an \textit{exact} conversion from one bipartite state to another under LOSR.

As previously stated, we use our technical results to look at embezzling and necessary communication for entangled state transformations. While motivated and focused on the work of van Dam and Hayden \cite{van-2003a} and Hayden and Winter \cite{Hayden-2003a}, there has been further related work. In terms of necessary communication, \cite{Aubrun-2008a} showed necessary and sufficient conditions for there to exist an exact catalytic transformation between two quantum states under LOCC. In terms of embezzlement, \cite{Leung-2011a} was the first to construct universal embezzling families with respect to zero communication for $m > 2$ parties. \cite{Leung-2014a} studied properties of universal embezzling families with respect to zero communication in more detail. 

In quantum thermodynamics, there is the notion of `thermal embezzling,' which is understood when the states commute with the relevant Hamiltonian, thereby reducing the problem to majorization \cite{Brandao-2015a,Ng-2015a}. \cite{Brandao-2015a} sketches how quantum embezzlement as done in van Dam and Hayden \cite{van-2003a} implies a type of thermal embezzlement that shows the second law of thermodynamics doesn't hold in an approximate catalytic sense. This formal insight is equivalent to what we call `randomness embezzlement' in Section \ref{sec:single-copy}. In \cite{Ng-2015a}, they can exclude this issue with the second law by imposing physical assumptions. They also construct an optimal universal thermal embezzling family. 

Finally, we remark that embezzling (with respect to any set of allowed transformations) is a form of approximate catalytic transformation. We refer the reader to recent reviews on catalysis for further information \cite{Datta-2023a,Lipka-Bartosik-2023a}.
}

\subsection{Organization of the Paper}
The rest of the paper is organized as follows. In Sections \ref{sec:notation} and \ref{sec:background} we present the necessary notation and background respectively to understand the rest of the paper. In Section \ref{sec:single-copy}, we
\begin{itemize}[itemsep=0mm]
\item Make explicit the correspondence between pure states under LU and the probability simplex and note this implies the existence of a classical variation of embezzlement (Proposition \ref{prop:classical-embezzling})
\item Prove our equation for fidelity of state conversion under local unitaries (Theorem \ref{thm:pure-state-conversions-with-LU}) and our optimization for fidelity of state conversion under local operations and shared randomness (Theorem \ref{thm:pure-state-conversion-via-LO})
\item Establish computable upper bounds on the fidelity of state conversion under LOSR  (Theorem \ref{thm:sdp-upper-bound}). 
\end{itemize}
In Section \ref{sec:many-copy} we present the results where the target or seed state is of i.i.d.\ form. In Section \ref{sec:catalyst-discussion} we discuss \edit{embezzlers} under local unitaries. In Section \ref{sec:on-extensions} we discuss why our theory does not generalize beyond bipartite pure states.

\section{Basic Notation}\label{sec:notation}
Our notation largely aligns with standard texts \cite{Wilde-2011a, Watrous-Book}. In this paper we consider finite dimensional quantum systems. Given $n \in \mbb{N}$, we define $[n] := \{1,...,n\}$. A finite dimensional Hilbert space will be labeled with a capital roman letter, e.g.\ $A,B,$ and may be identified by its dimension $d \in \mbb{N}$, e.g.\ $A \cong \mbb{C}^{d}$. The space of quantum states, or density matrices, with respect to a Hilbert space $A$, is the space of positive semidefinite operators with unit trace, i.e.\  $\Density(A) := \{\rho \in \Lin(A) : \rho \succeq 0 \, \& \Tr(\rho) = 1 \}$ where $\succeq$ is the L\"{o}wner order and $\Lin(A)$ is the space of endomorphisms. If a quantum state is a joint state over multiple Hilbert spaces, we will use a subscript to specify this, e.g. $\rho_{AB} \in \Density(A \otimes B)$. A quantum state $\rho_{A}$ is pure if $\rho_{A} = \dyad{\psi}$ for a unit vector $\ket{\psi}$, where we are using bra-ket notation. For this previous reason, we generally just specify a pure state by $\ket{\psi}_{A}$, or $\psi$ if we are considering its density matrix representation. We say a state is classical if it is diagonal in the standard basis, i.e. it's a probability distribution. The space of probability distributions over $d$ elements is denoted $\cP(d)$. A quantum channel $\cE$ is a (linear) completely positive, trace preserving map $\cE : \Lin(A) \to \Lin(B)$. \edit{We denote the set of quantum channels from $\Lin(A) \to \Lin(B)$ as $\Channel(A,B)$.}

\section{Background \& Motivation}\label{sec:background}
Throughout this section we fix $A \cong \mbb{C}^{d}, B \cong \mbb{C}^{d'}$ for clarity.

\paragraph{Fidelity}
The fidelity is a standard measure of similarity between two positive semidefinite operators $R,S \geq 0$.
\begin{equation}\label{eq:fid-def} F(R,S) = \left\|\sqrt{R}\sqrt{S} \right\|_{1}^{2} = \Tr\left(\sqrt{\sqrt{S}R\sqrt{S}}\right)^{2} \ ,
\end{equation}
where the square root of a positive semidefinite operator is defined in the standard fashion on its spectral decomposition and $\| \cdot \|_{1}$ is the Schatten $1-$norm. It satisfies various properties that will be relevant for this work which we summarize here. All of these may be verified by direct calculation or by referring to standard texts.
\begin{proposition}[Summary of Fidelity Properties]\label{prop:fidel-props}
Let $\rho,\sigma \in \Density(A)$. The following hold:
\begin{enumerate}[itemsep=0mm]
	\item $0 \leq F(\rho,\sigma) \leq 1$ where the upper bound is saturated if and only if $\rho = \sigma$ and the lower bound saturates if and only if their images are orthogonal. 
	\item The fidelity is isometrically invariant, i.e. given isometry $V: A \to B$,
	$$ F(V\rho V^{\dagger}, V\sigma V^{\dagger}) = F(\rho,\sigma) \ . $$
	\item The fidelity satisfies data-processing. That is, for any quantum channel $\cE \in \Channel(A,B)$,
	$$ F(\rho,\sigma) \leq F(\cE(\rho),\cE(\sigma)) \ . $$
	\item If both states are pure, 
	$$ F(\dyad{\phi},\dyad{\psi}) = |\bra{\psi}\ket{\phi}|^{2} \ , $$
	and if one state is pure
	$$ F(\dyad{\phi},\sigma) = \bra{\phi}\sigma\ket{\phi} \ . $$
	\item If both states are classical, $P,Q \in \cP(d)$, then the fidelity reduces to the square of the Bhattacharyya coefficient:
	$$ \hspace{1cm} F(P,Q) = \left(\sum_{i \in [d]} \sqrt{p(i)q(i)} \right)^{2} = BC(p,q)^{2} \ , $$
	where $p(i) = P(i,i)$ and likewise for $Q$.
	\item Given pure states with the same eigenbasis and real amplitudes, $\ket{\psi} = \sum_{x} \sqrt{p(x)} \ket{x}$, $\ket{\phi} = \sum_{x} \sqrt{q(x)} \ket{x}$ , the fidelity reduces to the square of the Bhattacharyya coefficient of the probability distributions defined by the amplitudes:
	$$ F(\dyad{\phi},\dyad{\psi}) = BC(p,q)^{2} \ . $$
\end{enumerate}
\end{proposition}
We also note that in all of these definitions there is a pesky squaring that effectively we don't care about. For this reason we could define the square root fidelity:
$$ \sqrt{F}(R,S) := \sqrt{F(R,S)} \ . $$
Note the square root fidelity could be viewed as the quantum extension of the Bhattacharyya coefficient.

\paragraph{Norms}
In defining the fidelity we used the Schatten $1-$norm. More generally, there are the Schatten $p-$norms which for $X \in \Lin(A,B)$ may be defined as $\|X\|_{p} := \|\sigma(X)\|_{p}$ where $\sigma(X)$ is the ordered vector of singular values of $X$, $\sigma_{1}(X) \geq \sigma_{2}(X) \geq ... \geq \sigma_{\mrm{rank}(X)}(X)$ and it is being evaluated under the $L_{p}-$norm where $p \geq 1$. The infinity norm, $\infty-$norm, is $\lim_{p \to \infty} \|X\|_{p} = \|X\|_{\infty} = \max_{i} \sigma_{i}(X)$. The infinity norm was generalized to the Ky Fan $k-$norms $\|X\|_{(k)} := \sum \sigma_{i}(X)$ for $1 \leq k \leq \min\{d,d'\}$. The Ky Fan norms have relevance in measuring entanglement \cite{de-2011a}. A generalization of the Ky Fan and Schatten norms together is given by the $(k,p)-$norms \cite{Mudholkar-1985a}
\begin{equation}
	\|X\|_{(k,p)} := \left(\sum_{i \in [k]} \sigma_{i}(X)^{p}\right)^{1/p} \ ,
\end{equation}
which also have use in measuring entanglement of pure states \cite{Johnston-2012a}. Much like is common to do for the Schatten $p-$norms, we can extend the $(k,p)-$norms to $p > 0$ with the caveat they won't be norms as they won't in general satisfy subadditivity (the triangle inequality) for $p \in [0,1)$.

\paragraph{Entanglement Theory} A bipartite positive operator $\rho_{AB}$ is separable if there exists  $n \in \mbb{N}$, $p \in \cP(n)$, $\{\sigma_{A}^{i}\}_{i \in [n]} \subset \Density(A)$, and $\{\tau_{B}^{i}\}_{i \in [n]}$ such that
$$ \rho_{AB} = \sum_{i \in [n]} p(i) \sigma_{A}^{i} \otimes \tau_{B}^{i} \ . $$
Otherwise the state is entangled. As a pure state $\dyad{\psi}_{AB}$ is defined by a unit vector, this reduces to a pure state is separable, referred to product in this setting, if and only if there exists $\ket{\phi}_{A},\ket{\varphi}_{B}$ such that $\ket{\psi} = \ket{\phi}_{A} \otimes \ket{\varphi}_{B}$. While this is sufficient for determining if a bipartite pure state is entangled, there is also a notion of `how' entangled a state is in terms of Schmidt rank. Every bipartite pure state $\ket{\psi}_{AB}$ admits a unique (up to re-ordering) decomposition of the form
\begin{equation}\label{eq:Schmidt-decomp}
	\ket{\psi}_{AB} = \sum_{i \in [k]} \sqrt{p(i)} \ket{u_{i}}_{A} \otimes \ket{v_{i}}_{B} \ ,
\end{equation}
where $k = \max\{d,d'\}$, $p \in \cP(k)$ and $\{\ket{u_{i}}\}_{i \in [k]}, \{\ket{v_{i}}\}_{i \in [k]}$ are orthonormal bases of $A$ and $B$ respectively. We note that for non-unit vectors, such a decomposition still exists, it is just $p \in \mbb{R}_{\geq 0}$, i.e. is a non-negative vector rather than a probability distribution.

The $\sqrt{p(i)} > 0 $ terms are referred to as the Schmidt coefficients of the pure state. The Schmidt rank of $\ket{\psi}_{AB}$, $\SR(\ket{\psi}) = \mrm{supp}(p)$, i.e.\ the number of Schmidt coefficients. This may be viewed as a measure of entanglement in the sense that the Schmidt rank of a product state is 1 and the maximally entangled state $\ket{\Phi^{+}}_{\mbb{C}^{d}\mbb{C}^{d}} = \frac{1}{\sqrt{d}} \sum_{i} \ket{i}_{\mbb{C}^{d}} \ket{i}_{\mbb{C}^{d}}$ has Schmidt rank $d$. We define the set $\SR(d) := \{\ket{\psi} : \SR(\ket{\psi}) \leq d \}$, where we note this set is independent of the dimension the state is embedded in.

Lastly we note a particularly nice property of pure states, known as Uhlmann's theorem.
\begin{lemma}[Uhlmann's Theorem]\label{lem:Uhl-thm}
Given $R,T \in \Pos(A)$ and (possibly unnormalized) $\ket{\psi} \in A \otimes B$ such that $\Tr_{B}(\psi) = R$, then 
$$ F(R,T) = \max\{|\bra{\psi}\ket{\phi}|^{2} : \ket{\phi} \in A \otimes B \, , \, \Tr_{B}(\phi) = T \} \ . $$
\end{lemma}

\paragraph{No-Go Theorems, Embezzling, \& Motivation}
With the established background, we now present the previous results related to zero communication pure state transformations which we will discuss our results in relation to. The first is a lower bound on the number of qubits or classical bits necessary to convert between pure states \cite{Hayden-2003a}.
\begin{proposition}(\cite[Theorem 8]{Hayden-2003a})\label{prop:Hayden-Winter}
Consider a state transformation via \edit{an arbitrary} channel $\cE \in \Channel(A\otimes B,A\otimes B)$ from seed state $\ket{\phi}_{AB}$ to target state $\ket{\psi}_{AB}$ such that $F(\cE(\phi),\psi) \geq 1 - \ve$ \edit{where the channel $\cE$ is implemented using local operations and quantum communication}. Then, independent of any amount of entanglement assistance, for $\delta = \sqrt[8]{\ve}$, in the implementation of $\cE$, $q$ qubits \edit{must have been} exchanged where
\begin{equation}\label{eq:Hayden-Winter}
\begin{aligned}
q \geq & \frac{1}{2}\left[\Delta_{\delta}(\Tr_{B}(\dyad{\psi})) - \Delta_{0}(\Tr_{B}(\dyad{\phi})) \right] \\
& \hspace{4cm} + \log(1-\delta) \ , 
\end{aligned}
\end{equation}
where 
\begin{equation*}
	\begin{aligned}
 \exp(\Delta_{\ve}(P)) =  \min \, \, & \rank(\wt{P}) \cdot \lambda_{\max}(\wt{P}) \\
\text{s.t.} &  \Tr (\wt{P}) \geq 1 - \ve \\
&  \,  \wt{P} = \Pi P \Pi \, \\
 & \, [P, \Pi] = 0 \\
 & \, \Pi^{2} = \Pi \ .
	\end{aligned}
\end{equation*}
Moreover, the bound given in \eqref{eq:Hayden-Winter} holds for a necessary amount of classical communication by multiplying the R.H.S.\ by two.
\end{proposition}
While the above proposition is very powerful and implies two states with different Schmidt decompositions cannot be perfectly converted with zero communication, it is not sufficient in every scenario. In particular, the following example shows that in certain cases Proposition \ref{prop:Hayden-Winter} cannot eliminate \textit{any} state from being able to be converted to a given target state with relatively high fidelities.
\begin{example}[On the necessity of communication] \label{ex:necessity-of-comm}
Up to local unitaries, let the target state be $\ket{\psi} = \sqrt{0.54} \ket{00} + \sqrt{0.02} \ket{11} + \sqrt{0.44}\ket{22}$. Assume we are interested in a state transformation $\cE$ such that $F(\cE(\phi),\psi) = 0.99$, where $\ket{\phi}$ is the seed state. Then $\ve = 0.01$, so $\delta > 0.56$. Note $\Tr_{B}(\dyad{\psi}) = 0.54 \dyad{0} + 0.02\dyad{1} + 0.44\dyad{2}$. Then $\Delta_{\delta}(\Tr_{B}(\dyad{\psi})) = \log( |1| \cdot 0.44) < -1.18$, by removing the $0.02$ and $0.54$ eigenvalues. It may be shown \cite{Hayden-2003a} that $\Delta_{0}(\Tr_{B}(\dyad{\psi})) \geq 0$ and clearly $\log(1-\delta) < 0$. It follows that in this setting the R.H.S. of \eqref{eq:Hayden-Winter} is negative. Therefore, we have no proof from this bound that any transformation for any seed state which achieves this relatively high fidelity of $99\%$ requires any communication.
\end{example}

While the above example shows there are reasonably small tolerated errors $\varepsilon$ where Proposition \ref{prop:Hayden-Winter} is not helpful, when the tolerated error is sufficiently small, it will imply the need for communication. This sort of structure for sufficiently small $\varepsilon$ also appears when considering quantum embezzlement \cite{van-2003a}, which may be seen as a solution to Proposition \ref{prop:Hayden-Winter} implying communication is necessary. Quantum embezzlement in effect shows one can make pure state transformations with zero communication to any non-zero error if they have the right sufficiently large \edit{embezzling state}.
\begin{proposition}\label{prop:q-emb}(\cite{van-2003a}) Consider the family of \edit{embezzler} states  $\ket{\mu(n)}_{A'B'} = \frac{1}{\sqrt{H_{n}}} \sum_{j=1}^{n} \frac{1}{\sqrt{j}}\ket{j}_{A'}\ket{j}_{B'}$ where $H_{n} := \sum_{i=1}^{n} n^{-1}$ is the Harmonic number. For any $\ve > 0$ and target bipartite pure state $\ket{\psi}_{AB}$ with Schmidt rank $m$, for $n > m^{1/\ve}$ there exist unitaries $U_{AA'}, W_{BB'}$ such that
\begin{align*}
& F(U_{AA'} \otimes W_{BB'} (\ket{\mu(n)}_{A'B'}\ket{0}_{A}\ket{0}_{B}),\\
& \hspace{3cm} \ket{\mu(n)}_{A'B'} \otimes \ket{\psi}_{AB}) \geq 1 - \ve \ .
\end{align*}
Moreover, $U,W$ are in effect permutations on the joint Schmidt bases.
\end{proposition}
One can see quantum embezzlement implies a way to convert one pure state to another to non-zero error by picking a large enough \edit{embezzler} and then first `embezzling out' the original state (uncomputing $\ket{\phi}$ to $\ket{0}\ket{0}$ via embezzling) and then `embezzling in' the target state $\ket{\psi}$.

What is perhaps most remarkable about the above approach is that it was shown in the original work that even if we allow LOCC and a state dependent \edit{embezzler, $\ket{\zeta}$, the amount of entanglement $n$ in $\ket{\zeta}$ must scale proportionally to $\frac{1}{\log(n)}$ as the error $\ve$ becomes close to zero, but the universal embezzling family $\ket{\mu(n)}$ scales the same way in the entanglement $n$.} That is, as $\ve$ \edit{approaches zero}, \edit{embezzling using van Dam and Hayden's universal embezzling family} is effectively optimal. However, just as with the discussion pertaining to Proposition \ref{prop:Hayden-Winter}, it's clear embezzling isn't necessary for reasonable error levels in general. In fact, we show in the following example that for any non-zero error there exist states which can be converted without any catalyst \edit{or embezzler}.
\begin{example}[On the necessity of embezzling] \label{ex:on-necc-of-emb}
As noted, as $\ve \to 0$, embezzling is necessary. However, it is not in general clear at what point embezzling becomes necessary. This can be seen as follows. Consider $\ve \in (0,1)$ and two probability distributions $p,q \in \cP(m)$ such that the $\mrm{BC}(p,q)^{2} \geq 1 - \ve$. Define the seed state as $\ket{\phi} = \sum_{i \in [m]} \sqrt{p(i)} \ket{i}_{A} \ket{i}_{B}$ and the target state as $\ket{\psi} = \sum_{i \in [m]} \sqrt{q(i)} \ket{i}_{A} \ket{i}_{B}$. Then we have
$$ F(\dyad{\phi},\dyad{\psi}) = BC(p,q)^{2} \geq 1 - \ve \ , $$
where we have used Item 5 of Proposition \ref{prop:fidel-props}. Therefore, given $\ket{\phi}$, it requires \textit{no} communication or entanglement to generate $\ket{\psi}$ to error $\ve$. In fact, as we show later (Proposition \ref{prop:local-unitaries}), this will be true for converting the set of states with Schmidt coefficients defined via $p$ to the set of states with Schmidt coefficients defined via $q$ in general.
\end{example}

Given these two examples, we see that while these results give strong characterizations of pure state transformations with zero communication, neither the need for communication by Proposition \ref{prop:Hayden-Winter} nor  the optimality of Proposition \ref{prop:q-emb} when the error tends to zero give us a full understanding of this setting. It would therefore be of value to better understand this task, and this is what the rest of this work addresses.

\section{Single Copy Pure State Conversion with Zero Communication}\label{sec:single-copy}
Our primary goal of this section is to determine the minimal error of conversion between pure states with zero communication, which would resolve the gap presented in Example \ref{ex:necessity-of-comm}. To establish the minimal error of conversion, we will use the correspondence between the probability simplex and Schmidt coefficients under local unitaries (LU), which we establish in the following subsection. We also note that this implies the existence of a classical equivalent of embezzling, which we call randomness embezzling (Theorem \ref{prop:classical-embezzling}). This correspondence motivates the idea that the optimal fidelity of pure state conversion under local unitaries is simply re-ordering the Schmidt coefficients, which we in fact prove (Theorem \ref{thm:pure-state-conversions-with-LU}). We then use the local unitary result to establish a bounded but non-linear optimization program that determines the optimal achievable fidelity under conversion via local operations and shared randomness (LOSR), which does not require shared randomness (Theorem \ref{thm:pure-state-conversion-via-LO}). We end the section by discussing the relationship between the LU and LOSR strategies and introducing an SDP relaxation for efficiently establishing upper bounds on the achievable fidelity of pure state conversions under LOSR.

\subsection{Correspondence Under Local Unitaries between Schmidt Coefficients and the Probability Simplex}
In this subsection we establish the bijection between Schmidt coefficients, which define the equivalence classes of bipartite pure states under local unitaries, and the probability simplex. One reason for this is because the rest of the results of this work might be best seen as verifying that in the zero communication setting this correspondence is all that matters. Indeed, we will see this in the subsequent subsections which show that the minimal fidelity error of pure state transformations under zero communication will always be functions of only the Schmidt coefficients.

\begin{definition}
We define the the ordered probability simplex, $\cP^{\downarrow}(d)$ as $p \in \cP^{\downarrow}$ if $p(i) \geq p(i+1)$ for all $i \in [d-1]$.
\end{definition}

\begin{proposition}\label{prop:local-unitaries}
Up to local unitaries, any pure quantum state is of the form 
$$ \ket{\psi}_{AB} = \sum_{i \in [k]} \sqrt{p^{\downarrow}(i)} \ket{i}_{A} \otimes \ket{i}_{B} \ , $$
where $p^{\downarrow}(i) \geq p^{\downarrow}(i+1)$ for all $i \in [k-1]$, $k = \max\{d,d'\}$, $p^{\downarrow} \in \cP^{\downarrow}(k)$, and $\{\ket{i}\}$ is the computational basis in both cases. In other words, there exist both equivalence classes on pure states under local unitary operations in terms of Schmidt coefficients and ordered Schmidt coefficients.
\end{proposition}
\begin{proof}
Consider $\ket{\psi}_{AB} = \sum_{j \in [k]} \sqrt{p'(j)} \ket{u_{j}} \otimes \ket{v_{j}}$ as decomposed in \eqref{eq:Schmidt-decomp}. Now fix the permutation $\pi$ on $[k]$ such that $p'(\pi^{-1}(i)) \geq p'(\pi^{-1}(i+1))$ for all $i \in [k-1]$, i.e.\ $\pi$ re-labels $p'$ so that it is decreasing. Define the unitaries $U_{A} = \sum_{j \in [k]} \ket{\pi(j)}\bra{u_{j}}$, $W_{B} = \sum_{j \in [k]} \ket{\pi(j)}\bra{v_{j}}$, which may be verified to be unitaries by direct calculation. Then $(U_{A} \otimes W_{B})\ket{\psi}_{AB}$ will be of the form given in the proposition statement. Finally, we could make this argument for any pure state without ordering the Schmidt coefficients to get one set of equivalence classes. As such, under local unitaries, we can define equivalence classes of pure states in terms of ordered or non-ordered Schmidt coefficients. This completes the proof.
\end{proof}

\begin{definition}
The space of (representatives of the equivalence class of) ordered Schmidt coefficient pure states with Schmidt rank bounded by $d$ is given by $\mrm{SR}^{\downarrow}(d)$. That is, if $\ket{\psi} \in \mrm{SR}^{\downarrow}(d)$, then $\ket{\psi} = \sum_{i \in [d]} \sqrt{p^{\downarrow}(i)} \ket{u_{i}} \ket{i}\ket{i}$ where $p^{\downarrow} \in \cP^{\downarrow}(d)$.
\end{definition}

We can use the previous proposition to relate the (ordered) probability simplex over $d$ elements to to the equivalence classes of (ordered) Schmidt decompositions with Schmidt rank bounded by $d$. This will make use of the vec mapping.

\begin{definition} Given the space of linear operators $\Lin(A,B)$ where $A \cong \mbb{C}^{d}$, $B \cong \mbb{C}^{d'}$, the vec mapping $\opvec: \Lin(A \otimes B) \to A \otimes B$ is defined by $\opvec(\ket{i}\bra{j}) = \ket{j} \otimes \ket{i}$ where $\{\ket{i}\}_{i \in [d]}$ and $\{\ket{j}\}_{j \in [d']}$ are the computational bases for $A$ and $B$ respectively.
\end{definition} 
This choice of definition for the vec mapping satisfies the identity 
\begin{equation}\label{eq:vec-map-ident}
	(X_{1}^{T} \otimes X_{0})\opvec(Y) = \opvec(X_0 Y X_1 ) \ ,
\end{equation}
where $X_0 \in \Lin(A_0 , B_0)$, $X_1 \in \Lin(A_1, B_1)$, and $Y \in \Lin(B_1 , B_0)$. 

\begin{proposition}
Consider the functions $\opvec(\sqrt{\cdot}): \Lin(\mbb{C}^{d}) \to \mbb{C}^{d} \otimes \mbb{C}^{d}$ and $\opvec^{-1}(\cdot^{\odot 2}): \mbb{C}^{d} \otimes \mbb{C}^{d} \to \Lin(\mbb{C}^{d})$ where $\cdot^{\odot 2}$ is the entry-wise square of a vector. These functions define a bijection between $\cP(d)$ (resp.\ $\cP^{\downarrow}(d)$) and the space of equivalence classes of Schmidt decompositions under local unitaries with Schmidt rank bounded by $d$ (resp.\ the space $\SR^{\downarrow}(d)$.)
\end{proposition}
\begin{proof}
We prove it via direct calculation for $\cP(d)$ and the space of Schmidt decompositions. The proof in the other case works the same. Let $C \cong \mbb{C}^{d}$. First, consider $p \in \cP(d)$ which we write in its density matrix form, e.g.\ $P = \sum_{i \in [d]} p(i) \dyad{i}$. Then 
\begin{align*} \opvec(\sqrt{P}) =& \opvec\left( \sum_{i \in [d]} \sqrt{p(i)} \dyad{i} \right) \\
=& \sum_{i \in [d]} \sqrt{p(i)} \ket{i}_{C} \otimes \ket{i}_{C'} \ ,
\end{align*}
which is in the specified equivalence class by applying an isometries that take the computational bases from $C,C'$ to $A,B$. In the other direction, take the Schmidt decomposition in the purified basis, $\ket{\psi}_{AB} = \sum_{i \in [d]} \sqrt{q(i)} \ket{i}_{A} \otimes \ket{i}_{B}$. We can convert the $A$ space to $C$ via the channel 
$$ \cF_{A \to C}(\cdot) := V^{\dagger} \cdot V + (\mbb{1} - V^{\dagger}V) \cdot (\mbb{1} - V^{\dagger}V) \ , $$
where $V = \sum_{i \in [d]} \ket{i}_{A}\bra{i}_{C}$ is the isometry that takes the $C$ space to the $A$ space as $|A| \geq |C|$ by assumption. The same type of conversion holds for the $B$ and $C'$. Therefore, we have (up to equivalences) $\ket{\psi}_{AB} = \sum_{i \in [d]} \sqrt{q(i)} \ket{i}_{C} \ket{i}_{C'}$. Then,
\begin{align*}
    \opvec^{-1}(\ket{\psi}^{\cdot 2}) =& \opvec^{-1}(\sum_{i \in [d]} q(i) \ket{i}_{C}\ket{i}_{C'}) \\
    =& \sum_{i \in [d]} q(i) \dyad{i}_{C} \ , 
\end{align*}
where in the last line we used that $C' \cong C$ so that $\Lin(C,C') \cong \Lin(C)$. This completes the proof.
\end{proof}
The reason this is useful is it draws equivalence between the equivalence classes of entangled states in terms of Schmidt coefficients and probability distributions under fidelity.
\begin{proposition}
Consider $\ket{\phi} = \sum_{i \in [d]} \sqrt{p(i)} \ket{i}_{A}\ket{i}_{B}$, $\ket{\psi} = \sum_{i \in [d]} \sqrt{q(i)} \ket{i}_{A}\ket{i}_{B}$. Then 
$ F(\dyad{\phi},\dyad{\psi}) = BC(p,q)^{2}$.
\end{proposition}
\begin{proof}
First note $V: \ket{i}_{A} \to \ket{i}_{A}\ket{i}_{B}$ is an isometry. \edit{Define $\ket{\phi'} := \sum_{i} \sqrt{p(i)}\ket{i}_{A}$ and similarly $\ket{\psi'}$. Note $V\ket{\phi'} = \ket{\phi}$ and $V\ket{\psi'} = \ket{\psi}$. Thus, we have
\begin{align*} 
F(\dyad{\phi},\dyad{\psi}) =& F(V\dyad{\phi'}V^{\dagger}, V\dyad{\psi'}V^{\dagger})\\
 =& F(\dyad{\psi'},\dyad{\phi'}) \\
 =& BC(p,q)^{2} \ , 
\end{align*}
where the first equality is our observation, the second is isometric equivalence (Item 2 of Proposition \ref{prop:fidel-props}), and the third is using Item 6 of Proposition \ref{prop:fidel-props} where we note $\ket{\psi'},\ket{\phi'}$ are of the form given in that item. This completes the proof.}
\end{proof}

\paragraph*{Randomness Embezzling}
Before moving forward, we note that independent of the focus of this work, this equivalence between Schmidt coefficients and the probability simplex means that the proof of quantum embezzlement also proves the existence of a classical version. Specifically, \edit{the proof of quantum embezzlement \cite{van-2003a} only bounds the fidelity between the Schmidt coefficients of the embezzling state and the target state tensored with the embezzling state (under some permutation). This allows them to bound the fidelity by reducing it to the Bhattacharyya coefficient of the Schmidt coefficients:}
\begin{align*}
F(\ket{\psi},\ket{\phi}) = |\langle \psi, \phi \rangle|^{2} =& \left|\langle \sqrt{P}, \sqrt{Q} \rangle \right|^{2} \\
=& \left(\sum_{i} \sqrt{p(i)q(i)} \right)^{2} \\ 
=& \mrm{BC}(p,q)^{2} \ \edit{.}
\end{align*}
\edit{This argument} follows the same \edit{form} as the previous few propositions. \edit{This allows us} to ultimately conclude the same proof bounds a classical equivalent of embezzling (Proposition \ref{prop:classical-embezzling}). \edit{This is shown in further formal detail in Appendix \ref{app:rand-embez-proof}.}

\edit{Moreover, we note this above idea has been previously explored in the context of quantum thermodynamics. Specifically, this exact idea was sketched in \cite{Brandao-2015a} where it was used to show that any process can be done cyclically when an approximate error condition is permitted and there are no restrictions on the embezzling state. Moreover, \cite{Ng-2015a} explored this issue in further detail and in particular showed thermal embezzling no longer violates the second law in an approximate sense when there are physical restrictions on the embezzling state.} As we did not present the proof for embezzlement of quantum states \edit{and the proof is omitted in \cite{Brandao-2015a},} we present the proof of embezzlement of probability distributions in full for clarity in Appendix \ref{app:rand-embez-proof}. 

\begin{figure}
\centering
\begin{subfigure}[b]{0.4\columnwidth}
    \begin{center}
    \includegraphics[width=\columnwidth]{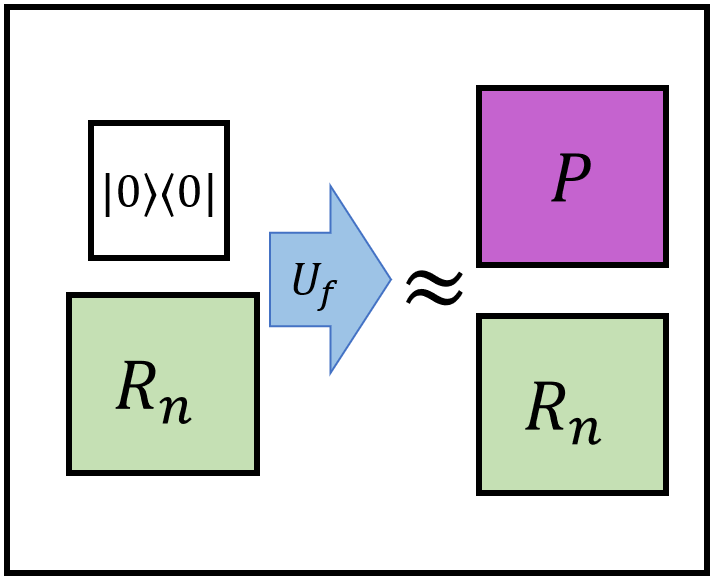}
    \end{center}
   \caption{}
\end{subfigure}
\begin{subfigure}[b]{0.54\columnwidth}
   \begin{center}
    \includegraphics[width=\columnwidth]{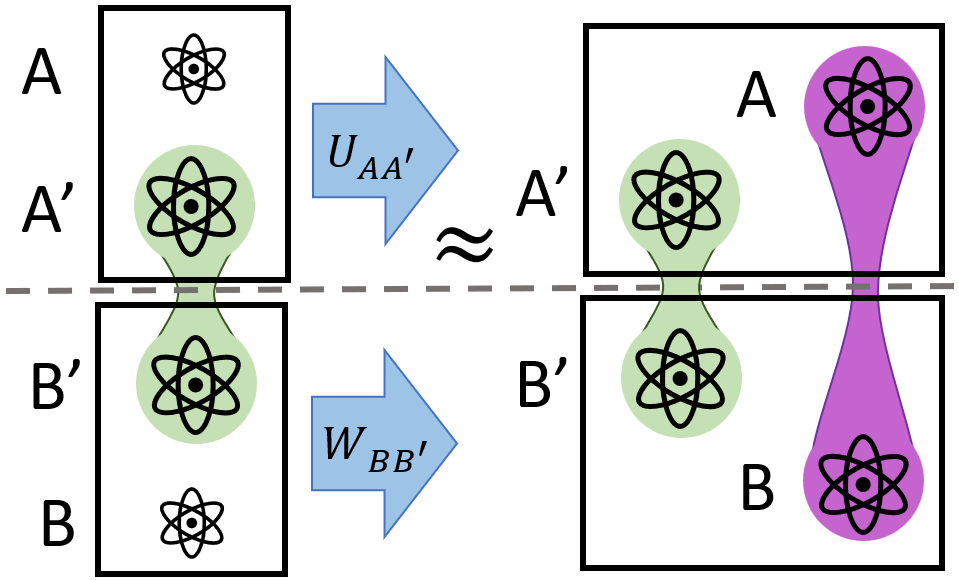}
    \end{center}
   \caption{}
\end{subfigure}
\caption{\small{Comparison between embezzlement of classical distributions and quantum states.
(a) The embezzlement of classical distributions happens within one lab and a local permutation of the joint computational basis. (b) The embezzling of quantum states happens across two labs where each party applies the permutation of the joint computational basis on their local halves. \edit{Note in both cases this is done in an approximate fashion as denoted by the $\approx$ symbol in the diagrams.}}}
\label{fig:classical-versus-quantum-embezzling}
\end{figure}

\begin{proposition}\label{prop:classical-embezzling}\edit{(See also \cite{Brandao-2015a,Ng-2015a})}
For any $\ve > 0$ and target probability distribution $P \in \cP(m)$, the \edit{embezzling} distribution $R_{n} := \frac{1}{H_{n}} \sum_{j=1}^{n} \frac{1}{j} \dyad{j}$ is such that for $n > m^{1/\ve}$ there exists a unitary representation of a basis relabeling $U_{f}$ of the joint distribution such that 
$$ F(U_{f}(R_{n} \otimes \dyad{0})U_{f}^{\dagger}, R_{n} \otimes P) \geq 1 - \ve \ . $$
\end{proposition}
We note the major difference between randomness and quantum embezzlement is the role of locality. In the classical case there is a single party and the distribution is not bipartite, both of which remove the notion of locality. These differences are non-trivial: one cannot construct a non-local classical equivalent of embezzling that at the same time demands that the \edit{embezzler} remains decoupled as in Proposition \ref{prop:q-emb}, and one cannot find a quantum equivalent of the non-local classical variation that one can implement as follows from Proposition \ref{prop:Hayden-Winter}. As it is not central to the rest of this work, we provide an extended discussion of this nuance for the interested reader in Appendix \ref{app:rand-embez-proof} after the proof of Theorem \ref{prop:classical-embezzling}. 

\subsection{Pure State Conversion under Local Unitaries}
Having established the relationship between the equivalence classes of pure states in terms of Schmidt coefficients and the probability simplex, we now show the optimal strategy for converting one pure state to another under local unitaries is simply re-labeling the Schmidt basis so the ordering of the Schmidt coeffficients is the same. This is not necessarily surprising. It is not clear what more one could do, and indeed this is the strategy that is used to implement quantum embezzlement \cite{van-2003a}. 

\begin{lemma}\label{lem:opt-fid-of-commuting-systems}
Let $R \in \Pos(\mbb{C}^{d})$, $T \in \Pos(\mbb{C}^{d'})$. Then 
$$ \max_{U} F(P,UQ U^{\dagger}) = F(\wt{R}^{\downarrow},\wt{T}^{\downarrow}) \ , $$
where $R^{\downarrow} = \sum_{i} \nu_{i}(R) \dyad{i}$, $\sigma_{1}(R) \geq n_{2}(R) \geq ... $ are the (decreasing) ordered eigenvalues of $R$, and likewise for $T^{\downarrow}$. In other words, the fidelity between $R$ and $T$ maximized over unitaries is equal to the fidelity of their ordered eigenvalues.
\end{lemma}
\begin{proof}
This proof is a combination of the definition of fidelity and a corollary of von Neumann's trace theorem. A similar identity was established in \cite{Berta-2010a}.
\begin{align*}
	& \max_{U} F(R,UTU^{\ast})\\
	 =& \max_{U} \|\sqrt{R}\sqrt{UTU^{\ast}}\|_{1}^{2}\\
	=& \max_{U} \|\sqrt{R}U\sqrt{T}U^{\ast}\|_{1}^{2}	\\
	=& \max_{U} \left( \max_{W} \left|\Tr[W\sqrt{R}U\sqrt{T}U^{\ast}] \right| \right)^{2} \\
	=& \max_{U,W} \left|\Tr[W\sqrt{R}U\sqrt{T}] \right| ^{2} \\
	=& \left( \sum_{i \in [q]} \sigma_{i}(R)\sigma_{i}(T) \right)^{2} \\
	=& F(\wt{R}^{\downarrow},\wt{T}^{\downarrow}) \ . 
\end{align*}
The first equality is \eqref{eq:fid-def}. The second equality is because $\sqrt{UTU^{\ast}} = \sum_{i} \sqrt{\lambda_{i}(T)} \dyad{\psi_{i}} = U\left(\sum_{i} \sqrt{\lambda_{i}(T)} \dyad{\phi_{i}}\right)U^{\ast} = U\sqrt{T}U^{\ast}$ by defining $\ket{\psi_{i}} := U \ket{\phi_{i}}$ where $\{\ket{\phi_{i}}\}_{i}$ is the eigenbasis of $T$. The third equality is a well-known variational form of the $1$-norm \cite{Watrous-Book}. The fourth is using cyclicity of trace, redefining $U^{\ast}W \to W$, and pulling out the maximization. The fifth is \cite[Corollary 7.4.1.3]{Horn-2012a}, and the final equality is by definition of the operators (which are defined in the same basis).
\end{proof}
We now can use the above lemma to establish the pure state property we are actually interested in. For notational simplicity, we define the following notation:
\begin{align}\label{eq:FLU-notation-definition}
	 F_{\mrm{LU}}(\rho,\sigma) := \max_{U,V} F(\rho,(U \otimes V)(\sigma)) \ , 
\end{align}
which is without loss of generality unitaries as we can just trivially embed the \edit{states $\rho,\sigma$ so that they both are defined on the same local spaces. That is, in general one would optimize over isometries, but by treating the states embedded into the same local spaces already, we can focus on unitaries. To see this, consider $\rho_{AB}$ and $\sigma_{A'B'}$ such that $|A| =: d_{A} > d_{A'} := |A'|$. We can then embed the local space $A'$ into $A$ by taking any bases of $A$ and $A'$ respectively, e.g. $\{\ket{\chi_{i}\}}_{i \in [d_{A}]}$ and $\{\ket{\zeta_{j}}\}_{j \in [d_{A'}]}$, and defining the isometry $W_{A' \to A} := \sum_{j \in [d_{A'}]} \ket{\chi_{i}}\bra{\zeta_{i}}$. We can then consider $W\sigma W^{\ast}$ instead of $\sigma$ directly. The same idea works for embedding $B'$ into $B$ when $|B'| \leq |B|$ and we can always treat $\rho$ as in the higher dimensional space (again, by embedding it). All of this is because it does not affect the eigenvalues of the respective/Schmidt coefficients of the relevant states. Finally, the choice of isometries that we do the embedding for does not matter as the $U,V$ we optimize over can take any choice of the local isometries on the $\sigma$ state to any others.}

\begin{theorem}\label{thm:pure-state-conversions-with-LU}
Let $\ket{\psi},\ket{\phi} \in A \otimes B$ be (possibly unnormalized) vectors with (possibly unnormalized) Schmidt coefficients $r_{1} \geq r_{2} \geq ...$, $t_{1} \geq t_{2} \geq ...$ respectively. Then,
\begin{equation}
 \max_{U,V} F(\ket{\psi},U \otimes V \ket{\phi}) = F(R^{\downarrow},T^{\downarrow}) \ ,
\end{equation}
where $R^{\downarrow} = \sum_{i} r_{i}\dyad{i}$, $T^{\downarrow} = \sum_{i} t_{i}\dyad{i}$. In particular, if we treat these as quantum states, this gives the optimal fidelity of converting $\ket{\phi}$ to $\ket{\psi}$ under local unitaries.
\end{theorem}
\begin{proof}
Up to local unitaries, $\ket{\psi} = \sum_{i} \sqrt{r_{i}} \ket{i}\ket{i}$. Therefore without loss of generality, that can be taken as our target state by allowing free local unitaries on the seed state. We can take the seed state to be of the form $\ket{\phi} = \sum_{i} \sqrt{t_{i}} \ket{i}\ket{i}$ by the same argument. Then by assumption, we are interested in $\max_{U,V} F(\ket{\psi},(U \otimes V)\ket{\phi})$ with the specified forms. Note 
\begin{align*}
& \Tr_{B}((U \otimes V)\dyad{\phi}(U\otimes V)^{\dagger})\\
=& \sum_{i,i'} \sqrt{t_{i}t_{i'}} U\ket{i}\bra{i'}U^{\dagger} \Tr(V\ket{i}\bra{i'}V^{\dagger}) \\
=& \sum_{i} t_{i} U\dyad{i}U^{\dagger} =: UQU^{\dagger}.
\end{align*}
Now for any unitary $U$ we define the following purification
\begin{align*}
    \ket{w^{|U}} := &\opvec(\sqrt{UQU^{\dagger}}) \\
    =& (\ol{U} \otimes U) \opvec(\sqrt{Q}) =(\ol{U} \otimes U)\ket{\phi} \ ,
\end{align*}
where we have used $\sqrt{UQU^{\dagger}} = U\sqrt{Q}U^{\dagger}$ and the vec map identity \eqref{eq:vec-map-ident}. Now we have 
\begin{equation}\label{eq:fid-pure-to-classical-step-1}
\begin{aligned}
F(R^{\downarrow},UQU^{\dagger}) =& \max_{\ket{w'}} F(\ket{\psi},\ket{w'})  \\
=& \max_{V} F\left(\ket{\psi},(\mbb{1} \otimes V)\ket{w^{|U}}\right) \\
=& \max_{V} F(\psi, (\ol{U} \otimes VU)\ket{\phi})  \ ,
\end{aligned}
\end{equation}
where the first equality is by  Uhlmann's theorem (Lemma \ref{lem:Uhl-thm}), the second is because all purifications of a given operator are unitarily equivalent on the purifying space \cite{Watrous-Book}, so there exists a $V$ such that $(\mbb{1} \otimes V)\ket{w^{|U}} = \ket{w'}$. The final line is just expanding the definition of $\ket{w^{|U}}$.

It follows,
\begin{align*}
    &\max_{W,V} F(\ket{\psi},(W \otimes V)\ket{\phi}) \\
    =&  \max_{\ol{U},V'} F(\ket{\psi}, (\ol{U} \otimes V'U) \ket{\phi}) \\
    = & \max_{\ol{U},V'} F(\ket{\psi}, (\mbb{1} \otimes V)\ket{w^{|U}}) \\
    = & \max_{U} F(R^{\downarrow},UQU^{\dagger}) \\
    = & F(R^{\downarrow},T^{\downarrow}) \ ,
\end{align*} 
where the first equality is because unitaries are closed under multiplication and the optimizations are independent, the second and third are both by \eqref{eq:fid-pure-to-classical-step-1} for clarity, the third is because unitaries are closed under conjugation and then the final equality is by applying Lemma \ref{lem:opt-fid-of-commuting-systems}. This completes the proof.
\end{proof}
This means under local unitaries, it is efficient to compute the optimal fidelity and that in fact the optimal strategy is simply Alice and Bob re-ordering the basis so that the Schmidt coefficients are in the same relative ordering. It also follows from Item 1 of Proposition \ref{prop:fidel-props} that unless all the Schmidt coefficients are equal, the fidelity cannot be one under local unitary strategies.

\subsection{Pure State Conversions under Local Operations and Shared Randomness}
While the previous section is nice in that it finds an efficient way of calculating the optimal conversion strategy under local unitaries, it would be natural to ask if local operations can do better than local unitaries as it is a much more general class of operations. In fact, we can see that it must do better in some cases in a trivial manner. Consider the target state $\ket{\psi}$ and the seed state $\ket{\phi} = \ket{\psi} \otimes \ket{\zeta}$ where $\ket{\zeta}$ is not product. Under local unitaries this transformation isn't possible to arbitrary precision because of $\ket{\zeta}$, but of course in reality the parties could trace out whichever portion(s) of $\ket{\zeta}$ they hold. Thus, we need a theory of transformations under local operations. 

Note that this trivial example we have given would not be resolved by local mixed unitary strategies\edit{, i.e. strategies where each party varies their choice of unitary according to some local randomness}. Indeed, we begin by noting that local mixed unitary strategies cannot ever outperform local unitary strategies.
\begin{corollary}
Let $\ket{\psi}$ be the target state and $\ket{\phi}$ be the seed state and only optimize over Alice and Bob using mixed unitary channels. Then the optimal is the same as in Theorem \ref{thm:pure-state-conversions-with-LU}.
\end{corollary}
\begin{proof}
Letting $\cE_{U},\cF_{W}$ be local mixed unitary maps,
\begin{align*}
& \max_{\cE_{U},\cF_{W}} F(\psi,(\cE_{U} \otimes \cF_{W})(\phi)) \\
=& \bra{\psi}(\cE_{U} \otimes \cF_{W})(\phi)\ket{\psi} \\
=& \int_{U,W} \bra{\psi}(U \otimes W)(\phi)\ket{\psi} dU \, dW  \\
\leq& \int_{U,W} \max_{U,W} \bra{\psi}(U \otimes W)(\phi)\ket{\psi}  \\
=& \max_{U,W} \bra{\psi} (U \otimes W)(\phi)\ket{\psi} \\ 
=& F(P^{\downarrow},Q^{\downarrow}) \ ,
\end{align*}
where the first equality is by Item 4 of Proposition \ref{prop:fidel-props}, the second is letting the mixed unitary map be for any probability measures $dU$,$dW$ over the unitary group. The inequality is because the inner product is real and so it is lower bounded by the maximum. The second to last equality is by linearity, and the final equality is by Theorem \ref{thm:pure-state-conversions-with-LU}. Noting that a specific choice of local unitaries is a special case of mixed unitary channels completes the proof.
\end{proof}
The above tells us that we must escape the use of unitaries to improve our bounds. Note however that in general the only maps that preserve pure states are isometries, and our results so far have been in terms of pure states, so we need to maintain this structure to build on them. For this reason, the following proof will make use of the isometric representation of quantum channels. 

For notational simplicity, we define the optimal fidelity of conversion under local operations and shared randomness (LOSR) fidelity
$$ F_{\mathrm{LOSR}}(\rho,\sigma) := \max_{\mu,\cE_{\lambda},\cF_{\lambda}} F(\rho,\int (\cE_{\lambda} \otimes \cF_{\lambda})(\sigma) d\mu(\lambda) ) \ , $$
where $\mu$ is a probability measure over an index set for sets of local channels $\{\cE_{\lambda}\}$ and $\{\cF_{\lambda}\}$. Similarly, we can define optimal fidelity of conversion under local operations (LO) as
$$ F_{LO}(\rho,\sigma) := \max_{\cE,\cF} F(\rho,\cE \otimes \cF)(\sigma)) \ . $$
With these defined, we prove the following.

\begin{theorem}\label{thm:pure-state-conversion-via-LO}
Let $\ket{\psi},\ket{\phi} \in A \otimes B$ be (possibly unnormalized) vectors with (possibly unnormalized) Schmidt coefficients $r_{1} \geq r_{2} \geq ...$ and $t_{1} \geq t_{2} \geq ...$ respectively. Then,
\begin{equation}
\begin{aligned}
& F_{LOSR}(\ket{\psi},\ket{\phi}) \\ 
=& F_{LO}(\ket{\psi},\ket{\phi}) \\
=& \max_{P' \in \cP(\Sigma)} F((R \otimes P')^{\downarrow},T^{\downarrow}_{\text{embed}}) \ ,
\end{aligned}
\end{equation} 
where \edit{the finite alphabet $\Sigma$ satisfies} $|\Sigma| \leq \edit{\Big[} \SR(\ket{\phi})\cdot \SR(\ket{\psi})\edit{\Big]}$, $R = \sum_{i} r_{i} \dyad{i}$ and similarly for $\edit{T}_{\text{embed}}$ \edit{which is} the distribution \edit{T} embedded into the joint \edit{probability simplex over the finite alphabet indexing $R$ and $\Sigma$}. In particular, if we treat $\ket{\psi},\ket{\phi}$ as quantum states, this gives the optimal fidelity of converting $\ket{\phi}$ to $\ket{\psi}$ under local operations and shared randomness.
\end{theorem}
\begin{proof}
The first equivalence follows similarly to the mixed unitary case. Clearly the class of LOSR strategies is more general than the class of LO strategies, so we just need to show LOSR is only as strong as LO here.
\begin{align*}
    F_{\mrm{LOSR}}(\phi,\psi)
    =&  F\left (\psi,\int (\cE_{\lambda} \otimes \cF_{\lambda})(\phi) d\mu(\lambda) \right) \\
    =& \int  \bra{\psi} (\cE_{\lambda} \otimes \cF_{\lambda})(\phi) \ket{\psi} d\mu(\lambda) \\
    \leq & \int \max_{\cE,\cF} \left[ \bra{\psi} (\cE \otimes \cF)(\phi)\ket{\psi} \right]  d\mu(\lambda) \\
    =& \max_{\cE,\cF} \bra{\psi}\cE \otimes \cF)(\phi)\ket{\psi} \\
    =& F_{LO}(\phi,\psi) \ ,
\end{align*}
where the first equality is by definition and denoting the optimizers by $\mu,\{\cE_{\lambda}\},\{\cF_{\lambda}\}$, the second is by linearity of the Lebesgue integral, the inequality is because $\bra{\psi} (\cE \otimes \cF)(\phi) \ket{\psi}$ is a real number for any choice of local channels, the third equality is because $\mu$ is a probability measure that is now independent of the argument of the integral, and the final equality is by definition. This proves the reduction of LOSR to LO if the target state is pure.

Next, we bound the dimension of $\Sigma$. We want to consider $\max_{\cE,\cF} F(\psi,(\cE \otimes \cF)(\phi))$. Without loss of generality, we assume the local spaces are `compressed' such that $d_{in} := \SR(\ket{\phi})$ so that $\cE,\cF$ both act on $\Lin(\mbb{C}^{d_{in}})$. We now show that without loss of generality we may restrict the output dimension of $\cE,\cF$ to be $d_{out} := \SR(\ket{\psi})$. This is just because we can project onto the support of the marginal of $\ket{\psi}$ on both local spaces, so we can restrict the local maps to this space. Formally, this can be seen as follows. Consider arbitrary $\cE,\cF$ and consider the target state (up to LU) $\ket{\psi} = \sum_{i} \sqrt{r_{i}} \ket{i}\ket{i}$. Define $\Pi_{P} := \sum_{i: r_{i} > 0}$, i.e. the projector onto the support of $\Tr_{B}(\psi) = \Tr_{A}(\psi)$, where the equality is up to the change in space. Note $\mrm{rank}(\Pi_{P}) = \mrm{Schmidt}(\psi)$. By construction, $(\Pi_{P} \otimes \Pi_{P}) \ket{\psi} = \ket{\psi}$.  Therefore,
\begin{align*}
 & F(\psi, (\cE \otimes \cF)(\phi)) \\
 =& \bra{\psi}(\cE \otimes \cF)(\phi)\ket{\psi} \\
 =& \Tr[\dyad{\psi} (\cE \otimes \cF)(\phi)] \\
 =& \Tr[\psi \Pi_{P}^{\otimes 2}(\cE \otimes \cF)(\phi)\Pi_{P}^{\otimes 2}] \ ,
\end{align*}
where in the first equality we have used Item 4 of Proposition \ref{prop:fidel-props} and the other two use cyclicity of trace along with invariance of $\psi$ under the projector.
Now we can expand, 
\begin{align*}
    & \Pi_{P}^{\otimes 2}(\cE \otimes \cF)(\phi)\Pi_{P}^{\otimes 2} \\
    =& \sum_{k,l} \Pi_{P}A_{k} \otimes \Pi_{P}B_{k} \phi A_{k}^{\dagger}\Pi_{P} \otimes B_{l}^{\dagger}\Pi_{P} \\
    \equiv& (\cE_{\Pi} \otimes \cF_{\Pi})(\psi) \ ,
\end{align*}
where $\{A_{k}\},\{B_{l}\}$ are the Kraus operators of $\cE,\cF$ respectively and $\cE_{\Pi},\cF_{\Pi}$ are CPTNI maps defined by $\{\Pi_{P} A_{k}\},\{\Pi_{P} B_{l}\}$ respectively. Note this equivalence holds as $(\Pi A_{k})^{\dagger} = A_{k}^{\dagger}\Pi_{P}$ since $\Pi_{P}^{\dagger} = \Pi_{P}$ so it is CP and it is TNI because
\begin{align*}
    \sum_{k} (\Pi_{P} A_{k})^{\dagger}(\Pi_{P} A_{k}) =& \sum_{k} A_{k}^{\dagger}\Pi_{P} A_{k}  \\
    \leq& \sum_{k} A_{k}^{\dagger} \mbb{1} A_{k} = \mbb{1} \ , 
\end{align*}
where we used $\Pi_{P}^{2} = \Pi_{P}$ in the first equality, $\Pi_{P} \leq \mbb{1}$ and that $\cE$ is CP in the inequality, and that $\cE$ is TP in the last inequality. An identical argument holds for $\cF_{P}$. This proves the optimizer is achieved with CPTNI maps  $\Trans(\Lin(\mbb{C}^{d_{in}}),\Lin(\mbb{C}^{d_{out}}))$. Finally, we can lift $\cE_{P},\cF_{P}$ to being CPTP, denoted $\widehat{\cE},\widehat{\cF} \in \Trans(\Lin(\mbb{C}^{d_{in}}),\Lin(\mbb{C}^{d_{out}}))$ by adding one Kraus operator, e.g.\ for $\cE_{P}$ add the Kraus operator $Z \in \Lin(\mbb{C}^{d_{in}},\mbb{C}^{d_{out}})$ where $Z^{\dagger}Z = (\mbb{1} - \sum_{k} A_{k}^{\dagger}\Pi A_{k}) \geq 0$ which always exists by definition of the space of positive semidefinite operators. By linearity, 
\begin{align*}
    F(\psi,(\cE \otimes \cF)(\phi)) =& \Tr[\psi(\cE_{\Pi} \otimes \cF_{\Pi})(\phi)] \\
    \leq& \Tr[\psi(\widehat{\cE} \otimes \widehat{\cF})(\phi)]  \ .
\end{align*}
Therefore, without loss of generality, the optimal channels are $\cE,\cF \in \Channel(\mbb{C}^{d_{in}},\mbb{C}^{d_{out}})$. Note this means that $\mrm{Rank}(J_{\cE}) \leq d_{in} d_{out}$ and likewise for $J_{\cF}$.

We now derive the equation using the isometric representation of the channel \cite{Watrous-Book}.
\begin{align*}
& \max_{\cE,\cF} F(\psi,(\cE \otimes \cF)(\phi)) \\
=& \langle \psi, (\cE \otimes \cF)(\phi) \rangle \\
=& \max_{V_{1},V_{2}, \ket{\zeta}} \left|\bra{\psi}\bra{\zeta} (V_{1} \otimes V_{2})\ket{\phi} \right|^{2} \\
=& \max_{U_{1},U_{2},\ket{\zeta}} \left| \bra{\psi}\bra{\zeta} (U_{1} \otimes U_{2})\ket{\phi}\ket{0}_{E_{1}}\ket{0}_{E_{2}} \right|^{2} \\
=& \max_{U_{1}',U_{2}'
,\ket{\zeta_{p'}}} \left|\bra{\psi}\bra{\zeta_{p'}}(U_{1}' \otimes U_{2}')\ket{\phi}\ket{0}_{E_{1}}\ket{0}_{E_{2}}\right| \\
=& \max_{P'} F( (R \otimes P')^{\downarrow}, T^{\downarrow}_{\mrm{embed}}) \ ,
\end{align*}
where the second equality is because there exists an isometric representation of each channel which means $(V_{1} \otimes V_{2})\ket{\phi}$ is a pure state, so we can apply Uhlmman's theorem to find a purification of $\ket{\psi}$ that saturates the bound, but as $\ket{\psi}$ is already pure, any purification will be a product state with a unit vector. The third line is because we can always convert an isometry into a unitary on the appropriately large space. The fourth line means that $\zeta_{p'} = \sum_{i'} \sqrt{p'(i)} \ket{i}\ket{i}$, which can always be achieved by local unitaries on the $E_{1}$ and $E_{2}$ spaces, which result on new unitaries on the other side but the same maximum. The final equality is just using Theorem \ref{thm:pure-state-conversions-with-LU} and we write $T_{\mrm{embed}}$ to stress it is defined over the whole alphabet. Lastly, as we established bounds on the ranks of the local maps Choi matrices, we have bounds $E_{1},E_{2} \leq d_{in}d_{out}$, which justifies the maximum and tells us how large of a system we have to consider in the statement of the theorem.
\end{proof}

It is useful to see how this result works. It in effect shows the following equivalence of conversion distance when measured under fidelity \edit{
\begin{equation}
\begin{aligned}
    &\mrm{F}\left\{\ket{\phi} \xrightarrow[LO]{} \ket{\psi}]\right\} \\
    & \hspace{1cm} = \,
    \max_{\ket{\zeta}} \mrm{F}\left\{\ket{\phi} \xrightarrow[LU]{} \ket{\psi} \otimes \ket{\zeta}  \right\} \ ,
\end{aligned}
\end{equation}
where here $\mrm{F}\{a \to_{\cO} b \}$ denotes the optimal conversion between $a$ and $b$ using maps included in set $\cO$ according to fidelity. (See \cite{Gour-2024a} for a discussion of conversion distance in terms of trace norm for general resource theories.)} This statement can be viewed both by proof and via intuition as a special case of the isometric representation of a channel.
Moreover, it is easy to see in this form how it handles our motivating example. Indeed, if the target state is $\ket{\psi}$ and the seed state is $\ket{\psi} \otimes \ket{\zeta}$, then clearly the maximizer is chosen by the ancillary state being $\ket{\zeta}$ and the local unitaries being trivial. 

\begin{example}[Tightened Error Bounds for Zero Communication]\label{ex:tightened-error}
We now recall Example \ref{ex:necessity-of-comm}, which showed Proposition \ref{prop:Hayden-Winter} could not show that any seed state $\ket{\phi}$ would require communication to be mapped to $\ket{\psi} = \sqrt{0.54}\ket{00} + \sqrt{0.02}\ket{11} + \sqrt{0.44}\ket{22}$. Here we calculate the error if the seed is \edit{a qutrit} maximally entangled state\edit{, which we note is an example of the task of dilution.} By Theorem \ref{thm:pure-state-conversion-via-LO},
\begin{align*}
	F_{\mrm{LOSR}}(\ket{\psi},\ket{\Phi^{+}_{3}}) &= \max_{P' \in \cP([9])} F((P \otimes P')^{\downarrow},\frac{1}{3}\mbb{1}_{\mbb{C}^{3}}) \\
	&= \frac{1}{3} \max_{P' \in \cP^{\downarrow}([9])} F((P \otimes P')^{\downarrow},\mbb{1}_{\mbb{C}^{3}}) \ ,
\end{align*}
where we have used a slight abuse of notation as $\mbb{1}_{\mbb{C}^{3}}$ is embedded in a larger space. We note $P'$ can be ordered as it won't change $(P \otimes P')^{\downarrow}$. \edit{By the definition of fidelity \eqref{eq:fid-def},}
\begin{align*}
& F((P \otimes P')^{\downarrow},\mbb{1}_{\mbb{C}^{3}}) \\
=& \edit{\left( \max_{p' \in \cP^{\downarrow}([9])} \sqrt{0.54p'(1)} + \sqrt{0.44p'(1)} + \sqrt{0.54p'(2)} \right)^{2}}  \ ,
\end{align*}
where we have used $p\edit{'}(1) \geq p\edit{'}(2)$ and $0.54>0.44$, so these must be the largest 3 elements without loss of generality. As \edit{we are considering} a maximization, it can only be increased by letting $p\edit{'}(2) = 1-p\edit{'}(1)$. \edit{This allows us to parameterize in terms of $p' \in [0,1]$. Taking the derivative with respect to $p'$ and setting equal to zero, we find ${p'}^{\star} \approx 0.6014$}. Plugging this back in, we get 
$$ F_{\mrm{LOSR}}(\ket{\psi},\ket{\Phi^{+}_{3}}) \edit{\approx \frac{1}{3}(1.59018)^{2} \approx 0.8429} \ . $$
Thus, not only would communication be necessary, but without any, the fidelity can be far from one.
\end{example}

\subsection{Relation between LO and LU Strategies}
The natural question given the previous theorems is if we can better understand the relationship between LO and LU strategies. We first show that LU and LO strategies are equivalent when either the target or the seed state is a two qubit state. 

\subsubsection{LU and LO Equivalence for Two-Qubit Seed or Target State}

\begin{proposition}
Consider (possibly unnormalized) entangled two qubit seed state $\ket{\phi} \in \mbb{C}^{2} \otimes \mbb{C}^{2}$. Let the (possibly unnormalized) target entangled state be $\ket{\psi} \in \mbb{C}^{d} \otimes \mbb{C}^{d'}$. Then the optimal non-communicative strategy is the local unitary strategy.
\end{proposition}
\begin{proof}
First, we point out that it suffices to consider normalized distributions. This is because if we have unnormalized vectors, then, using the definition of fidelity,
$$ F((R\otimes P')^{\downarrow},T^{\downarrow}) = \Tr[R]\Tr[T]F((P\otimes P')^{\downarrow},Q^{\downarrow}) \ , $$
where $P := \Tr[R]^{-1}R$, $Q:= \Tr[T]^{-1}T$. It follows since these scaling factors will appear for both the LU and LO case (and thus cancel when comparing the values), the normalized case is sufficient.

Without loss of generality, $q^{\downarrow} = (q,1-q)$ where $q \geq 1/2$ and $p^{\downarrow} = (p(1),p(2),...)$. Then the optimal local unitary strategy is $\sqrt{q p(1)} + \sqrt{(1-q)p(2)}$. For any $P'$ we can write $(p')^{\downarrow} = (p'(1),p'(2),...)$. The optimal CPTP strategy (up to a square) is of the form 
$$ \sqrt{qp(1) p'(1)} + \sqrt{(1-q) \max\{p(1)p'(2),p(2)p'(1)\}} \ . $$
These values can only increase by assuming $p'$ has two outcomes, so let us assume so without loss of generality and parameterize the distribution by $p' \in [1/2,1]$ to obtain 
$$\sqrt{qp(1)p'} + \sqrt{(1-q)\max\{p(1)(1-p'),p(2)p'\}} \ . $$
Moreover note $p(2)p' < p(2)$ unless $p' = 1$, which is equivalent to the LU strategy, so the second entry in the maximization would be lower than the LU setting. Therefore, we focus on the remaining case. We are specifically interested in when the following strict inequality holds:
\begin{align*}
& \sqrt{qp(1)p'} + \sqrt{(1-q)p(1)(1-p')} \\
& \hspace{2cm} > \sqrt{qp(1)} + \sqrt{(1-q)p(2)} \\
\Leftrightarrow \, & g(p') := \sqrt{qp(1)}(\sqrt{p'}-1) \\
& \hspace{1.5cm} + \sqrt{1-q}(\sqrt{p(1)(1-p')} \\
& \hspace{4cm} - \sqrt{p(2)}) > 0 \ .
\end{align*}
Then $\frac{d}{dp'}g(p') = \frac{\sqrt{qp(1)}}{2\sqrt{p'}} + \frac{\sqrt{p(1)(1-q)}}{2\sqrt{1-p'}}$. It follows,
\begin{align*}
    & \frac{\sqrt{qp(1)}\sqrt{1-p'}}{2\sqrt{p'}\sqrt{1-p'}} + \frac{\sqrt{p'}\sqrt{p(1)(1-q)}}{2\sqrt{1-p'}\sqrt{p'}} \geq 0 \\
    \Leftrightarrow & \sqrt{qp(1)}\sqrt{1-p'} + \sqrt{p'}\sqrt{p(1)(1-q)} \geq 0 \\
    \Leftrightarrow & \sqrt{q}\sqrt{1-p'} + \sqrt{p'}\sqrt{(1-q)} \geq 0 \\
    \Leftrightarrow & \sqrt{F}(Q^{\downarrow},P'^{\downarrow}) \geq 0 \ ,
\end{align*}
where the first line is multiplying to get identical denominators, the second line is multiplying by the denominator, the third is dividing out $p(1)$, and the final is by the definition of square root fidelity. Note the final inequality will always hold strictly unless $q \in \{0,1\}$, i.e. the state is a product state, by Item 1 of Proposition \ref{prop:fidel-props}. If $q \in \{0,1\}$, then the state is a product state which would contradict that we assume the state is entangled. Therefore, in our setting, $g(p')$ only increases over its interval, $p' \in [0,1]$. Thus, the optimal choice of $p'$ is $p' = 1$, but in this case the value is $\sqrt{qp(1)} \leq \sqrt{qp(1)} + \sqrt{(1-q)p(2)}$, i.e.\ the optimal choice is lower bounding the optimal local unitary strategy. It follows this is never optimal. This completes the proof.
\end{proof}

\subsubsection{LU and LO Inequivalence for States with Schmidt Rank Greater than Two}
If there is equivalence for two qubit seed or target states, it is natural to ask if this property persists. One might expect that this is a special property of qubit systems as are found throughout quantum information science results. Indeed, generally this property does not hold, which we will prove via example.
\begin{theorem}
For seed and target state with Schmidt rank $\geq 3$, the optimal LO strategy may be better than the optimal LU strategy.
\end{theorem}
\begin{proof}
We construct an example for Schmidt rank $3$. By continuity of the fidelity, one can embed the target and seed in bigger spaces with arbitrarily small perturbations for it to hold in higher dimensions, which is why this is sufficient. Consider target state $\ket{\psi} = 0.85\ket{00}+ 0.08\ket{11} + 0.07\ket{22}$ and seed state $\ket{\phi} = 0.45(\ket{00} + \ket{11}) + 0.1\ket{22}$. Then, the optimal LU strategy fidelity is
\begin{align*} 
    & F(P^{\downarrow},Q^{\downarrow})\\
    =&  \left(\sqrt{0.45}(\sqrt{0.85}+\sqrt{0.08}) + \sqrt{0.1(0.07)}\right)^{2} \\
    <& 0.796 \ .
\end{align*}
In contrast, if we consider $P' = [0.55,0.28,0.17]$, then 
\begin{align*}
& F((P \otimes P')^{\downarrow},Q^{\downarrow}) \\
=& \left(\sqrt{0.45}\sqrt{0.4675} + \sqrt{0.45}\sqrt{0.238} + \sqrt{0.1}\sqrt{0.1445} \right)^{2} \\
>& 0.82 \ .
\end{align*}
As we maximize over $P'$, the optimal LO strategy achieves a value that is strictly above the LU strategy. This completes the proof.
\end{proof}

\subsection{Inefficiency of Optimal LOSR Fidelity and Computable Upper Bounds}
In the above we have constructed an example where the local operations strategy outperforms the local unitary strategy (though we have not shown what the strategy itself is). A natural question would then be how easy it is to solve for the optimal fidelity value or even a bound. By Theorem \ref{thm:pure-state-conversions-with-LU}, we can conclude the optimal local unitary strategy is polynomial time to solve as all one needs to do is sort the Schmidt coefficients and calculate the fidelity. Indeed, one could solve for the ordering of the Schmidt coefficients using the linear program for sorting a vector. 

In contrast, for optimizing LO strategies, we have no such luck. In effect this is because there are two things to optimize over at once. Indeed, recall from Theorem \ref{thm:pure-state-conversion-via-LO} that
$$ F_{LO}(\ket{\psi},\ket{\phi}) = \max_{P' \in \cP(\Sigma)} F((P\otimes P')^{\downarrow},Q^{\downarrow}) \ . $$
Then the problem is that one must first tensor $P$ onto variable $P'$ and then re-order the vector. One cannot even in general order an optimization variable, which we will refer to as `sorting,' as sorting is in general non-convex. In sorting a vector using a linear program, one relaxes to bistochastic channels and considers a linear function so that the optimizer is an extreme point which by the Birkhoff von Neumann theorem is a specific permutation. However, we are many levels of involvement above that: we want the distribution $P'$ such that its product distribution $P \otimes P'$ when sorted optimizes the fidelity with $Q^{\downarrow}$. Therefore, we need to optimize over $P'$ and the permutation at the same time. It's not clear that we can actually relax to bistochastic strategies because of the joint concavity of fidelity. That is to say, for any bistochastic channel $\cE$,
\begin{align*}
    F(\cE(P \otimes P'),Q^{\downarrow}) =& F(\sum_{\pi} r(\pi) V_{\pi}(P \otimes P'),Q^{\downarrow}) \\
    \geq & \sum_{\pi} F(r(\pi)V_{\pi}(P \otimes P'),r(\pi)Q^{\downarrow}) \\
    =& \sum_{\pi} r(\pi) F(V_{\pi}(P \otimes P'),Q^{\downarrow}) \ ,
\end{align*}
where the first line is Birkhoff-von Neumann theorem, the second is joint concavity using $Q^{\downarrow} = \sum_{\pi} r(\pi) Q^{\downarrow}$ as $r$ is a probability distribution, and the last line is because $F(\lambda P,Q) = \lambda F(P,Q) = F(P,\lambda Q)$. Thus any bistochastic channel \textit{may} strictly do better than the average of its extreme points. Moreover, even if we could optimize over bistochastic channels, we would have a non-convex objective function as the bistochastic channel, an optimization variable, would be applied to $P \otimes P'$ which is also partially an optimization variable.

Given the above, it seems likely the best option if one were to try and find a (near) optimum would be to use gradient descent from random initial $P'$, realizing it will only work locally and will break down at `kinks' where the ordering changes. Otherwise more sophisticated non-convex optimization techniques might be used.

\paragraph*{Computable Upper Bound Methods}
Perhaps even worse than our inability to calculate the exact fidelity, is that it is not clear in general how to determine good bounds. Certainly we have the following result.
\begin{theorem}
Unless the target state is \edit{(up to local unitaries)} $\ket{\psi}=\ket{\phi} \otimes \ket{\zeta}$ for some pure state $\ket{\zeta}$ where $\ket{\phi}$ is the seed state, there exists $\varepsilon > 0$ such that there does not exist local operations that will take $\ket{\phi}$ to $\ket{\psi}$.
\end{theorem}
\begin{proof}
\edit{Theorem \ref{thm:pure-state-conversion-via-LO} states that 
$$ F_{\text{LOSR}}(\ket{\psi},\ket{\phi}) = \max_{P' \in \cP(\Sigma)} F((R \otimes P')^{\downarrow}, T^{\downarrow}_{\text{embed}}) \ . $$ 
Item 1 of Proposition \ref{prop:fidel-props} states that fidelity between two normalized states (and thus distributions) is one if and only if the two arguments are the same. Thus, the maximization obtains one if and only if there exists $P' \in \cP(\Sigma)$ such that $(R \otimes P')^{\downarrow} = T^{\downarrow}_{\text{embed}}$. This means the Schmidt coefficients of $\ket{\phi}$ are $\{r_{i}p'_{k}\}_{(i,k)}$. That is, (up to local unitaries) $\ket{\phi} = \sum_{(i,k)} r_{i}p'_{k} \ket{(i,k)}\ket{(i,k)}$. Defining local isometry $U_{A \to A_{0}A_{1}} \ket{(i,k)}_{A} =\ket{i}_{A_{0}}\ket{k}_{A_{1}}$, we have
\begin{align*}
&(U\otimes U)\ket{\psi} \\
 =& \sum_{i,k} r_{i}p'_{k} \left(\ket{i}_{A_{0}}\ket{k}_{B_{0}}\ket{i}_{A_{1}}\ket{k}_{B_{1}}\right) \\
 =& \left(\sum_{i} r_{i}\ket{i}_{A_{0}}\ket{i}_{B_{0}} \right) \otimes \left(\sum_{k} p_{k}\ket{k}_{A_{1}}\ket{k}_{B_{1}} \right) \\
 =:& \ket{\psi} \otimes \ket{\zeta} \ ,
\end{align*}
where we defined $\ket{\zeta}$ as the second state. This completes the proof.}
\end{proof}
The above theorem, while derived from a very different strategy than Proposition \ref{prop:Hayden-Winter}, does not seem to give us much more information as to at what point communication is necessary. What we would want to efficiently improve this would be to establish upper bounds on the equation given in Theorem \ref{thm:pure-state-conversion-via-LO} that have a closed form that does not depend on $P'$. One option is to use the data processing inequality for fidelity. This can be seen in the following proposition.
\begin{proposition}
Consider (possibly unnormalized) target state $\ket{\psi}$ and seed state $\ket{\phi}$ with corresponding Schmidt distributions $p,q$ respectively. If $p_{\max} \leq q_{\max}$, then 
$$ F_{LO}(\ket{\psi},\ket{\phi}) \leq F(\ol{p},\ol{q}) \ , $$
where $\ol{p} = p_{\max} \dyad{0} + (1-p_{\max})\dyad{1}$ and likewise for $\ol{q}$.
\end{proposition}
\begin{proof}
Without loss of generality let $d$ be the maximum local dimension. Let $\cE(\cdot) =  \dyad{0} \cdot \dyad{0} + \sum_{i \in \{1,...,d-1\}} \ket{1}\bra{i} \cdot \ket{i}\bra{1}$. That is, $\cE$ coarse-grains a probability distribution to the Bernoulli distribution with its first element untouched and the sum of all the others as the other outcome. Then using data processing of fidelity (Item 3 of Proposition \ref{prop:fidel-props}),
\begin{align*}
    & \max_{P' \in \cP(\Sigma)} F((P \otimes P')^{\downarrow},Q^{\downarrow}) \\
    \leq & \max_{P' \in \cP(\Sigma)} F(\cE((P \otimes P')^{\downarrow}),\cE(Q^{\downarrow})) \\
    = & \max_{p' \in [0,1]} F\left(\wt{P}(p'), \cE(Q^{\downarrow})\right) \ ,
 \end{align*}
 where $\wt{P}(p') := p_{\max}p' \dyad{0} + (1-p_{\max}p')\dyad{1}$ and $\cE(Q^{\downarrow}) = q_{\max}\dyad{0}- (1-q_{\max})\dyad{1}$. Now note that by assumption $p_{\max} \leq q_{\max}$. As the fidelity will only decrease as $p_{\max}p'$ moves away from $q_{\max}$, the optimal choice is $p' = 1$. This completes the proof. 
\end{proof}
The problem with the above bound is that there will be cases where $p_{\max} > q_{\max}$. Why the inequality in the other direction was required was to know for a fact what element of $p$ was relevant, namely $p_{\max}$ and that any choice of $p' \neq 1$ would be sub-optimal. In general this strategy would require $q^{\downarrow}(j)$ is sufficiently large relative to $p^{\downarrow}(j)$. This can be determined in some cases. Here we provide a simple example.
\begin{example}
Let 
\begin{align*}
p^{\downarrow} = [3/4,1/8,1/8]^{\Trans} \quad q^{\downarrow} = [1/2,1/2]^{\Trans} \ .
\end{align*}
Then $(p \otimes p')^{\downarrow}[1:2] = p'(1)[3/4, 1/8]^{\Trans}$, and so we can coarse-grain on the second element to obtain $\ol{P}(p') = 1/8p'\dyad{0}+ (1-1/8p')\dyad{1}$ and $\ol{Q} = Q^{\downarrow}$. Then as $1/8p' < 1/2$, the upper bound is $F(\frac{1}{8} \dyad{0} + \frac{7}{8}\dyad{1}, \frac{1}{2}\mbb{1}) \approx 0.83$.
\end{example}

The above shows that while data processing can be sufficient in certain cases, it does not provide an easy general method. Another common alternative in quantum information theory is semidefinite relaxations of optimization problems because semidefinite programs are efficient to evaluate. In Appendix \ref{app:SDP-relaxation}, we establish the following upper bound and show it may be expressed as a semidefinite program, which, as everything is in terms of probability distributions, is due to the non-linearity of fidelity and nothing particularly quantum.
\begin{theorem}\label{thm:sdp-upper-bound}
Consider (possibly unnormalized) target state $\ket{\psi}$ and (possibly unnormalized) seed state $\ket{\phi}$. Let $\mrm{SR}(\psi) = d$ and $\mrm{SR}(\phi) = d'$. Define $A = \mbb{C}^{d}$, $B = \mbb{C}^{d \cdot d'}$. Then,
\begin{equation}\label{eq:SDP-upper-bound}
    \begin{aligned}
    F_{\mrm{LOSR}}(\ket{\psi},\ket{\phi}) \leq \max & \; F(R,Q^{\downarrow}_{\text{embed}}) \\
    \mrm{s.t.} & \; \Tr_{B }[R] = P^{\downarrow} \\
    & \; R \in \cP^{\downarrow}(d^{2} \cdot d') \ ,
    \end{aligned}
\end{equation}
where $P$ and $Q$ are the distributions defined by $\ket{\psi}$ and $\ket{\phi}$'s Schmidt coefficients respectively. Moreover, this admits the following simple semidefinite program over the reals:
\begin{equation}\label{eq:Reals-SDP-representation}
    \begin{aligned}
        \max & \sum_{i \in [d^{2} \cdot d']} x(i) \\
        \text{s.t.} & \; \begin{pmatrix} \text{diag}(r) & \text{diag}(x) \\ \text{diag}(x) & \text{diag}(q^{\downarrow}_{\text{embed}}) \end{pmatrix} \succeq 0 \\
        & \Tr_{B}[\text{diag}(r)] = P^{\downarrow} \\
        & r \in \cP^{\downarrow}([d^{2} \cdot d]) \\ 
        & x \in \mbb{R}^{d^{2}\cdot d'} \ ,
    \end{aligned}
\end{equation}
\end{theorem}
Physically, this relaxation may be seen as relaxing the isometric representation of the optimal LOSR strategy to one where one allows the ancillary environment start off entangled with the local system. Mathematically, this is not too loose because we require this entangled pure state has a notion of ``local Schmidt coefficients" that pertain to the original target state, although this physically does not seem to have a clean interpretation. Nonetheless, we can see that \eqref{eq:SDP-upper-bound} will not achieve unity unless there exists a joint distribution $Q=R$, which would require $Q^{\downarrow}_{\text{embed}}$ to have $P^{\downarrow}$ as it's marginal, which seems highly restrictive. Therefore, \eqref{eq:SDP-upper-bound} should provide an upper bound that is non-trivial.

\section{Many Copy Pure State Conversion with Zero Communication}\label{sec:many-copy}
Having established what happens for single copies, we consider many copies. We provide two motivations for doing this. First, we note that it's not clear what the limiting behaviour will be even in the LU setting. A reader may recall from other works that the fidelity is multiplicative so if $F(P,Q) < 1$, then $\lim_{n \to \infty} F(P^{\otimes n},Q^{\otimes n}) = \lim_{n \to \infty} F(P,Q)^{n} \to 0$. However, we lose the multiplicativity as we are considering $\lim_{n \to \infty} F((P^{\otimes n})^{\downarrow},(Q^{\otimes n})^{\downarrow})$. This issue is further aggravated if we consider local operations and the ancillary variable. 

The second motivation is that what was initially considered in the literature, albeit with LOCC \cite{Bennett-1996a}, was the conversion of many copies of states. A particular focus in the referenced work and subsequent ones is the case where either the target or seed state is the maximally entangled state, known as distillation and dilution respectively. With LOCC, we know there are `rates' in the conversions. By \cite{Hayden-2003a} along with previous results in this work, we would not expect there to be non-negative rates without the communication assuming the error is required to be vanishing, i.e.\ $\ve \to 0$.

In this section we establish convex optimization problems for dilution and distillation in the zero communication setting. These results are established in terms of the not-actually-a-norm $\|\cdot \|_{(k,1/2)}$, which we remind the reader is the $(k,p)-$norms extended to $p < 1$ introduced in Section \ref{sec:background} with the choice of $p = 1/2$. We also look at the limiting behaviour as the number of copies grows. In particular, we find a closed form when trying to convert $n-$fold two qubit states to a different $n-$fold two qubit state. Moreover, we prove the fidelity goes to zero in this case. We discuss the extension of this to entangled states with larger Schmidt rank. 

\subsubsection{Dilution Under Local Operations}
We begin by determining the limits of dilution. For intuition, we begin with local unitaries where there is no optimization. Recall that the Schmidt coefficients of the maximally entangled state are all $\sqrt{d^{-1}}$, so they correspond to the maximally mixed distribution under our bijection between Schmidt coefficients and probability distributions.
\begin{proposition}
For local unitary strategies the optimal dilution fidelity is given by
$$F_{LU}\left(\ket{\psi}, \ket{\Phi^{+}_{d}}^{\otimes n}\right) = d^{-n} \left\|P\right\|_{(d^{n},1/2)} \ . $$
\end{proposition}
\begin{proof}
Generally, if $\ket{\psi} \neq \ket{\Phi^{+}_{d}}$,
\begin{align*}
1 >& F(P^{\downarrow},\pi_{d}^{\otimes n})^{\downarrow})  \\
=& F(P^{\downarrow},\pi_{d}^{\otimes n}) \\
=& \left[ d^{-n/2} \sum_{i \in [d^{n}]} \sqrt{P^{\downarrow} (i)} \right]^{2} \\
=& d^{-n} \left[ \sum_{i \in [d^{n}]} \sqrt{P^{\downarrow} (i)} \right]^{2} \\
=& d^{-n} \|P\|_{(d^{n},1/2)} \ , 
\end{align*}
where the first equality is because $\pi^{\otimes n}_{d}$ is invariant under ordering, the second is using the definition of fidelity and that $\pi^{\otimes n}_{d}$ has uniform coefficients, and the final equality is the definition of the $(k,p)-$norms. In particular note we have dropped the sorting.
\end{proof}
We remark we could have set $\ket{\phi} = \ket{\phi'}^{\otimes m}$ to get a tradeoff, but this does not seem to provide any insight.

Just as in the one-shot setting, we know the above result isn't as useful in general because it can't throw out resources, so we now present the general result.
\begin{proposition}
The optimal fidelity of converting $n$ $d-$local dimensional EPR pairs to $\ket{\psi}$ under local operations is given by 
$$ F_{LO}(\ket{\psi},\ket{\Phi^{+}_{d}}^{\otimes n}) = d^{-n} \max_{P' \in \cP(\Sigma)} \|(P \otimes P')\|_{(d^{n},1/2)} \ , $$
where $\|\cdot\|_{(k,p)}$ is $(k,p)-$norm generalized to $p \geq 0$.
Moreover, for fixed $n$, this is a convex optimization problem.
\end{proposition}
\begin{proof}
Starting from the result of Theorem \ref{thm:pure-state-conversion-via-LO},
\begin{align*}
    & \max_{P' \in \cP(\Sigma)} F((P \otimes P')^{\downarrow}, (\pi_{d}^{\otimes n})^{\downarrow}) \\
    =& \max_{P' \in \cP(\Sigma)} F((P \otimes P')^{\downarrow}, \pi_{d}^{\otimes n}) \\
    =& \left[\frac{1}{d^{n/2}} \max_{P' \in \cP(\Sigma)} \sum_{i \in [d^{n}]} \sqrt{(P \otimes P')^{\downarrow}(i)} \right]^{2} \quad (\star) \\
    =& \frac{1}{d^{n}} \max_{P' \in \cP(\Sigma)} \|P \otimes P'\|_{(d^{n},1/2)} \ ,
\end{align*}
the first inequality is invariance of $\pi_{d}^{\otimes n}$ under sorting, the second is definition of fidelity and that each element of $\pi_{d}^{\otimes n}$ is the same, the last is the definition of $(k,p)$-norm extended to $p \geq 0$. 

To show this is a convex optimization problem, note that $\Phi_{P}(\cdot) := P \otimes \cdot$ is linear, $-\sqrt{\cdot}$ is operator convex, and the sum of the $k$ largest eigenvalues of a PSD $P$, which we will denote $\Sigma_{k}(P)$ is convex. Thus, starting from $(\star)$,  
\begin{align*}
& \left[d^{-n/2} \max_{P' \in \cP(\Sigma)} \sum_{i \in [d^{n}]} \sqrt{P \otimes P'}^{\downarrow}(i)\right]^{2} \\
=& \left[-d^{-n/2} \min_{P' \in \cP(\Sigma)} \mathlarger{\mathlarger{\Sigma}}_{d^{n}}\left(-\sqrt{\Phi_{P}(P')}\right)\right]^{2} \ ,
\end{align*}
where we have used $\max_{x \in \cC} f(x) = - \min_{x \in \cC} -f(x)$ and our definitions. Then ignoring the $-d^{-n/2}$ factor and the square, the optimization problem is over the probability simplex, which is a convex subset of the positive semidefinite matrices, and the objective function is convex over the positive semidefinite cone as $-\sqrt{\Phi_{P}(\cdot)}$ is operator convex and $\Sigma_{d^{n}}$ is a convex function over the space of Hermitian operators. this completes the proof.
\end{proof}
Unfortunately, while this gives computable bounds, it is not clear how one could determine the optimal value analytically.

\subsubsection{Distillation Under Local Operations}
We now present the same results in the distillation case,
where we take some state to many EPR states. For completeness, we state the local unitaries case.
\begin{proposition}
The fidelity of distillation under local unitaries and zero communication is given by
$$ F_{LU}(\ket{\Phi^{+}_{d}}^{\otimes m},\ket{\psi}^{\otimes n}) = d^{-m} \|P^{\otimes n}\|_{|S|,1/2} \ , $$
where $S = [\min\{d^{m},\mrm{rank}(P)^{n}\}]$.
\end{proposition}
\begin{proof}
The proof is effectively identical to the dilution case by symmetry of the fidelity.
\end{proof}
In contrast to the local unitary case, the symmetry is broken when one considers local operations.
\begin{theorem}
For fixed $d,m,n$ the optimal fidelity for \edit{distillation} under local operations is given by
\begin{align*} 
& F_{LO}(\ket{\Phi^{+}_{d}}^{\otimes m},\ket{\psi}^{\otimes n}) \\
=& d^{-m} \left[ \min_{P' \in \cP^{\downarrow}(\Sigma)} - \sum_{i \in \cI} \alpha_{i}\sqrt{p'(i)} \right]^{2} \ ,
\end{align*}
where $\cP^{\downarrow}(\Sigma)$ is the set of decreasing distributions as defined in Section \ref{sec:background},  $\cI \equiv [\lceil \rank(P)^{n}/d^{m} \rceil]$, and $\alpha_{i} := \sum_{j \in [(i-1)d^{m}: \min\{i\cdot d^{m},\rank(P)^{n}\}]} \sqrt{p^{\downarrow}_{n}(i)}$. Note the minimization is a convex optimization program.
\end{theorem}
\begin{proof}
Yet again, we use the square root fidelity and then take the square at the end. Then, using Theorem \ref{thm:pure-state-conversion-via-LO}, we have
\begin{align*}
& F_{LO}((\Phi^{+}_{d})^{\otimes m},\psi^{\otimes n}) \\
=& \max_{P' \in \cP(\Sigma)} F((\pi_{d}^{\otimes m} \otimes P')^{\downarrow},(P^{\otimes n})^{\downarrow})  \\
=& \left[ \max_{P' \in \cP(\Sigma)} \sum_{i \in S} \sqrt{(\pi^{\otimes m}_{d} \otimes p')^{\downarrow}(i)}\sqrt{p_{n}^{\downarrow}(i)} \right]^{2} \ .
\end{align*} 
Next, note 
$$(\pi^{\otimes m}_{d} \otimes P')^{\downarrow} = d^{-m/2} \sum_{i' \in \Sigma}  p^{\downarrow}(i') \mbb{1}_{\mbb{C}^{d^{m}}} \ , $$
where we have just used that $\pi_{d}^{\otimes m}$ is invariant under ordering. It follows that if we let $\cI \equiv [\lceil \rank(P)^{n}/d^{m} \rceil]$, we can rewrite,
\begin{align*}
& F_{LO}((\Phi^{+}_{d})^{\otimes m},\psi^{\otimes n}) \\
=& d^{-m} \Bigg[ \max_{P' \in \cP(\Sigma)} \sum_{i \in \cI} \sqrt{(p')^{\downarrow}(i)} \\
& \hspace{1.5cm} \cdot \sum_{j \in [(i-1)d^{m}: \min\{i\cdot d^{m},\rank(P)^{n}\}]} \sqrt{p^{\downarrow}_{n}(i)} \Bigg]^{2} \ .
\end{align*}
Now first define $\alpha_{i} := \sum_{j \in [(i-1)d^{m}: \min\{i\cdot d^{m},\rank(P)^{n}\}]} \sqrt{p^{\downarrow}_{n}(i)}$ as these coefficients may be pre-computed. Second, note that the probability simplex restricted to descending distributions, $\cP^{\downarrow}(\Sigma)$ is itself convex as $r_{\lambda}^{\downarrow} := \lambda p^{\downarrow} + (1-\lambda)q^{\downarrow}$ satisfies 
\begin{align*} 
\lambda p^{\downarrow}(i) + (1-\lambda) q^{\downarrow}(i) \geq \lambda p^{\downarrow}(i+1) + (1-\lambda) q^{\downarrow}(i) \ ,
\end{align*}
for all $i \in [|r|]$. Thus we have, 
\begin{align*}
& F_{LO}(\ket{\Phi^{+}_{d}}^{\otimes m},\ket{\psi}^{\otimes n}) \\
=& \Bigg[ -d^{-m} \min_{P' \in \cP^{\downarrow}(\Sigma)} -\sum_{i \in \cI} \alpha_{i}\sqrt{p'(i)} \Bigg]^{2} \ .
\end{align*} 
The minimization is a convex optimization problem because if we consider $f(p'):= -\sum_{i}\alpha_{i}\sqrt{p'(i)}$, then its Hessian is $\nabla^{2} f = \sum_{i} [\alpha_{i}/4 p'(i)^{-3/2}] \dyad{i}$, which is positive semidefinite on the interior of the probability simplex (i.e. when $p'(i) > 0$ for all $i$). This completes the proof.
\end{proof}

\subsubsection{Two Qubit Setting}
We have now seen that even in the basic dilution and distillation setting, while we can determine convex optimization programs, we can't seem to get clean analytic results. In this section we consider an even more tractable setting to attempt to resolve this: many copy two-qubit seed and target states. We show in this setting under certain assumptions the local unitary strategy is optimal and lobby this to show in particular that the optimal fidelity of converting $n$ copies of $\ket{\phi}$ to $n$ copies $\ket{\psi}$ goes to zero as $n$ goes to infinity. We note that this setting is more manageable because we effectively only have to reason about Bernoulli distributions.
\begin{lemma}
Given Bernoulli distribution $P = p\dyad{0} + (1-p)\dyad{1}$, then $P^{\otimes n}$ is such that the sequence $x^{n}$ with $(n-k)$ zeros has probability $p^{n-k}(1-p)^{k}$. Moreover, there are $\binom{n}{k}$ sequences with probability $p^{k}(1-p)^{n-k}$ and the same for $p^{n-k}(1-p)^{k}$.
\end{lemma}
\begin{proof}
The claim that $x^{n}$ with $(n-k)$ zeros has probability $p^{n-k}(1-p)^{k}$ is straightforward. The second point actually just follows from the fact there are $\binom{n}{k}$ sequences with $k$ zeros, which could be proven by induction in a straightforward manner.
\end{proof}
We can now use the above lemma along with Theorem \ref{thm:pure-state-conversions-with-LU} to get the optimal LU fidelity as a function of the number of copies $n$.
\begin{corollary}\label{corr:two-qubit-LU-F}
Consider entangled states $\ket{\psi},\ket{\phi} \in \mbb{C}^{2} \otimes \mbb{C}^{2}$. Then,
\begin{align*} 
    & F_{LU}(\psi^{\otimes n},\phi^{\otimes n}) \\
    & \hspace{0.5cm} = \sum_{k \in [n]} \binom{n}{k} (pq)^{(n-k)/2} ((1-p)(1-q))^{k/2} \ .
\end{align*}
\end{corollary}
\begin{proof}
By Theorem \ref{thm:pure-state-conversions-with-LU} we can reduce to the Bernoulli distributions from the Schmidt coefficients, $\ket{\psi}^{\otimes n} \mapsto P^{\otimes n}$, $\ket{\phi}^{\otimes n} \mapsto Q^{\otimes n}$. Since these are Bernoulli distributions, if we assume without loss of generality $p \geq (1-p)$, we can order the probabilities simply by the exponent, e.g. $p^{j-k}(1-p)^{k} \geq p^{j-k-k'}(1-p)^{k+k'}$ for any $0 \leq k' \leq j-k$. Moreover, the cardinality of each set of sequences will be the same for both $P^{\otimes n}$ and $Q^{\otimes n}$ because $\ket{\psi},\ket{\phi}$ are only entangled if their Schmidt rank is two. Therefore,
\begin{equation}\label{eq:Bernoulli-expansion}
\begin{aligned}
    & F((P^{\otimes n})^{\downarrow},(Q^{\otimes n})^{\downarrow}) \\
    =& \sum_{k \in [n]} \binom{n}{k} (pq)^{(n-k)/2} ((1-p)(1-q))^{k/2} \
\end{aligned}
\end{equation}
where the sum is over the number of zeros in the string, the cardinality was proven in the previous lemma, and the last term is just a re-writing of $\sqrt{p^{n-k}(1-p)^{k}}\sqrt{q^{n-k}(1-q)^{k}}$.
\end{proof}
 We note it is straightforward to generalize the above result to the case where you have the number of states differs between the seed and the target, but the form would be ugly as one would need to count how many sequences of a given probability there are and keep track of this in the sum. Indeed at this point the problem is elaborate enough that there is no advantage with dealing with two-qubit states as it's a question of the type classes \cite{Cover-2006}. We state this as a remark.
\begin{remark}
Consider states $\ket{\psi},\ket{\phi}$ respectively with ordered probability distributions corresponding to their Schmidt coefficients, $P$ and $Q$ respectively. $F_{LU}(\ket{\psi}^{\otimes n},\ket{\phi}^{\otimes m})$ can be computed. This is because the probability of a given sequence drawn in i.i.d. form from a distribution has a closed form \cite[Theorem 11.1.2]{Cover-2006}. It follows that as long as one determines the type classes exactly and takes into account that the sizes of the type classes may differ between $P$ and $Q$, the computation is possible, albeit tedious.
\end{remark}

Rather than dealing with the computational nightmare of generalizing beyond two qubit states, we now show that the term in Corollary \ref{corr:two-qubit-LU-F} always goes to zero as $n$ goes to infinity.
\begin{proposition}
Consider entangled states $\ket{\psi},\ket{\phi} \in \mbb{C}^{2} \otimes \mbb{C}^{2}$.
$$ \lim_{n \to \infty} F_{LU}(\ket{\phi}^{\otimes n},\ket{\psi}^{\otimes n}) = 0 \ . $$
\end{proposition} 
\begin{proof}
 Let the probability distributions corresponding to their Schmidt coefficients be parameterized by $p \edit{\geq 1/2}$ and $\edit{1/2 \leq} q = p + \ve$ where $\ve \in [-1/2,1/2]$. \edit{Therefore, using Corollary \ref{corr:two-qubit-LU-F}, we have}
\begin{equation}\label{eq:Bernoulli-Schmidt-parameterization}
\begin{aligned}
& F_{LU}(\ket{\psi}^{\otimes n},\ket{\phi}^{\otimes n}) \\
=& \sum_{k \in [n]} \binom{n}{k} (p^{2} + \edit{p}\ve)^{(n-k)/2} \\
& \hspace{2cm} \cdot [(1-p)^{2} -\ve(1-p)]^{k/2} \ .
\end{aligned}
\end{equation}
Now note $p^{2} + \edit{p}\ve \edit{= p q} < 1$ as otherwise \edit{both states would be product}. Define $\alpha := (p^{2}+\edit{p}\ve)^{1/2} < 1$. Then we have
\begin{align*}
    & \binom{n}{k} (p^{2} + \edit{p}\ve)^{(n-k)/2}  \cdot [(1-p)^{2} -\ve(1-p)]^{k/2} \\
    \leq & \left(\frac{n \cdot e}{k} \right)^{k} \alpha^{n-k} [(1-p)^{2} -\ve(1-p)]^{k/2} \\
    =& \left(\frac{e}{k} \alpha^{-1}\right)^{k}[(1-p)^{2} -\ve(1-p)]^{k/2} n^{k} \cdot \alpha^{n} \\
    =& O(\mrm{poly}(n)) O(\mrm{exp}(-n)) \\
    \to & 0 \ ,
\end{align*}
where \edit{the inequality uses a standard} upper bound on the binomial coefficient, in the first equality we have grouped terms by scaling. \edit{The second equality uses asymptotic notation, where we remind the reader $g(n) = O(f(n))$ means for sufficiently large $n$, $g(n) \leq C f(n)$ for some constant $C$ and $\mrm{poly}(n)$ (resp.~$\mrm{exp}(n)$) denote the sets of functions polynomial (resp.~exponential) in integer $n$. To make this conversion, we} have used that \edit{everything but $\alpha^{n}$} is a polynomial in $n$ and that $\alpha < 1$, so $\alpha^{n}$ scales inverse exponentially in $n$. The limiting factor is then because an inverse exponential times a polynomial goes to zero. We also remark that the term where $k = n$ will also go to zero as $[(1-p)^{2} - \ve(1-p)]^{k/2}$ will go to zero as $k$ goes to infinity as its magnitude will be bounded by 1. 

Therefore, each term in the sum goes to zero as $n$ goes to infinity, so the entire sum will go to zero. This completes the proof.
\end{proof}
We note our proof tells us nothing about the scaling as a function of the difference between $p$ and $q$ nor does it tell us how fast it goes to zero compared to $F(P^{\otimes n},Q^{\otimes n})$. These are shown numerically for specific cases in Fig. \ref{fig:fidelity-under-LU-plots}.
\begin{figure}
\centering
   \includegraphics[width=\columnwidth]{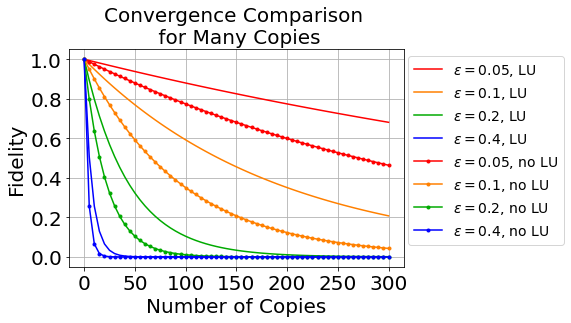}
\caption{\small{\edit{Comparison of $F(\ket{\psi}^{\otimes n},\ket{\phi}^{\otimes n})$ and $F_{LU}(\ket{\psi}^{\otimes n},\ket{\phi}^{\otimes n})$ as a function of how different the Schmidt coefficients are and as a function of the number of copies. Here we consider many copies of $\ket{\psi}$ $=$ $\sqrt{0.55}\ket{0}\ket{0}$ $+$ $\sqrt{0.45}\ket{1}\ket{1}$ being converted into the same number of copies of $\ket{\phi}$ $=$ $\sqrt{q}\ket{0}\ket{0}$ $+$ $\sqrt{1-q}\ket{1}\ket{1}$ where $q := p + \varepsilon$ for choices of $\ve$ denoted in the legend. Dotted lines denote when there are no local unitaries applied to re-order the Schmidt coefficients where as straight lines are the case where the local unitary strategy is applied.}}}
\label{fig:fidelity-under-LU-plots}
\end{figure}

It is then natural to ask if what we have seen so far is something special to local unitaries. We show that under sufficient conditions, just like in the single copy case, when two-qubit seed states are involved, local unitary strategies are optimal.
\begin{theorem}
Let $\ket{\psi} \in \mbb{C}^{2} \otimes \mbb{C}^{2}$ and the target state be $\ket{\psi}^{\otimes n}$. Let the seed state $\ket{\phi}$ satisfy $\SR(\ket{\phi}) \leq n\SR(\ket{\psi})$. Then the optimal local operations strategy is the optimal local unitary strategy.
\end{theorem}
\begin{proof}
By Theorem \ref{thm:pure-state-conversion-via-LO},
\begin{align*}
 & F_{LO}(\ket{\psi}^{\otimes n},\ket{\phi}) \\
 =& \max_{P' \in \cP(\Sigma)} F((P^{\otimes n} \otimes P')^{\downarrow},Q^{\downarrow}) \\ 
 =& \sum_{i \in |Q|} \sqrt{Q^{\downarrow}(i)} \sqrt{ (P^{\otimes n} \otimes P')^{\downarrow}(i)} \ .
\end{align*}
We will show that $P'$ should be the delta distribution. If $p \neq 1/2$, $p'(1) < 1$, then for any $0 \leq k \leq n$, we have the inequalities
\begin{align*}
    p^{n-k}(1-p)^{k} >& p^{n-k}(1-p)^{k}p'(1)\\
    >& p^{n-k}(1-p)^{k} p'(2)
\end{align*}
and
\begin{align*}
    p^{n-k}(1-p)^{k} >& p^{n-k}(1-p)^{k}p'(1) \\
    >& p^{n-(k+1)}(1-p)^{k+1} p'(1) \\
    >& p^{n-(k+1)}(1-p)^{k+1} p'(2)
\end{align*}
As square root is a monotone, this holds when we take the square root. Note that by assumption $P^{\otimes n}$ has enough entries by itself for there to be one corresponding to each $q^{\downarrow}$. Therefore, given the inequalities above, it follows if $p'(1) \neq 1$, each term in the sum only decreases. Therefore, $p'(1)$ is optimal for every $n$ and $k$. Thus, when $p \neq 1/2$, the optimal value is obtain by $P'$ being a delta distribution, which means it's equivalent to the local unitary strategy. \\

Finally, if $p = 1/2$, then $p^{n-k}(1-p)^{k} = 2^{-n}$ for all $k$. Therefore, if $p'(1) < 1$, the inequalities simplifies for all $0 \leq k \leq n$:
\begin{align*}
 p^{n-k}(1-p)^{k}p'(1) =& p^{n-(k+1)}(1-p)^{k+1}p'(1)\\
>& p^{n-(k+1)}(1-p)^{k+1}p'(2)
\end{align*}
and
\begin{align*}
 p^{n-k}(1-p)^{k}p'(1)> p^{n-k}(1-p)^{k}p'(2) \ .
\end{align*}
Again because each $q$ term is paired up already, this means if $p'(1) \neq 1$, the value decreases. Therefore, we again conclude the optimal strategy is the LU strategy. This completes the proof.
\end{proof}
We note that a trivial example of why we need the Schmidt rank constraint in the previous theorem is our original example for the advantage of LO strategies: if $\ket{\phi}^{\otimes n + \ell}$ where $\ell \geq 1$, then there is a better LO strategy than an LU strategy. Finally, we note it immediately follows from these previous results that
\begin{corollary}
If $\ket{\phi},\ket{\psi} \in \mbb{C}^{2} \otimes \mbb{C}^{2}$ are both entangled, then
$$ \lim_{n \to \infty} F_{LO}(\ket{\psi}^{\otimes n},\ket{\phi}^{\otimes n}) = 0 \ . $$
\end{corollary} 

\section{On Generalized Embezzling Conversion}
\label{sec:cat-conv}\label{sec:catalyst-discussion}
\FloatBarrier

We now have established a rather robust theory of pure state transformations under local operations. It is natural to return to the topic of conversion of one state to another using an ancillary entanglement, i.e.\ catalytic transformations \edit{and their embezzling relaxation.} Of course, it is immediate from our results so far that we know the optimization program that determines the optimal pure \edit{embezzling} state as we state in the following proposition.
\begin{proposition}
For any Schmidt rank $d$, the optimal pure state \edit{embezzler} for state conversion $\ket{\phi}$ to $\ket{\psi}$ is the quantum state $\ket{\zeta} = \opvec(\sqrt{R})$ that is determined via the optimization
\begin{align}
	 \max_{R \in \cP(d),\mrm{P'}\in \cP(\Sigma)} F((P \otimes P')^{\downarrow},(Q \otimes R)^{\downarrow}) \ .
\end{align}
\end{proposition}
\begin{proof}
This immediately follows from the input being $\ket{\phi} \otimes \opvec(\sqrt{R})$ and then applying Theorem \ref{thm:pure-state-conversion-via-LO}. Note this means $|\Sigma|$ scales as function of $d$.
\end{proof}
However, as we have already addressed, even without a free variable for the \edit{embezzling state}, the optimization in Theorem \ref{thm:pure-state-conversion-via-LO} seems unmanageable directly. While in principle one could use the relaxation in Theorem \ref{thm:sdp-upper-bound} to obtain efficient upper bounds, it is less obvious how often these will be non-trivial given that $R$ is a free variable.

The next most natural setting would be that of \edit{approximate} catalytic state conversion under local unitaries, i.e.\ we consider transformations of the form
\begin{align}
	\ket{\phi}\ket{\zeta} \edit{\approx_{\ve}^{LU}} \ket{\psi}\ket{\zeta} \ ,
\end{align}
where $\ket{\zeta}$ is \edit{an embezzling state and $\approx_{\ve}^{LU}$ denotes reversible equivalence under local unitary transformation up to error $\ve$ as measured under fidelity.} This may be seen as a generalization of \edit{traditional} embezzlement where $\ket{\phi} = \ket{0}_{A}\ket{0}_{B}$ and $\ket{\zeta} = \ket{\mu(n)}$.\footnote{We refer the reader to Proposition \ref{prop:q-emb} if the notation has been forgotten.}

Now as noted in the background, embezzling is known to be in effect optimal for sufficiently small $\varepsilon$. It follows for sufficiently small error $\varepsilon > 0$, the strategy that embezzles out the seed state and then embezzles in the target state is roughly optimal, i.e.\
\begin{align}
 \ket{\phi}\ket{\mu(n)} \edit{\approx_{\ve}^{LU}} \ket{0}\ket{0}\ket{\mu(n)} \edit{\approx_{\ve}^{LU}} \ket{\psi}\ket{\mu(n)}
\end{align}
is effectively optimal \edit{where we remind the reader $\ket{\mu(n)}$ is the van Dam-Hayden embezzling family $n$ pertains to the Schmidt rank of the given state in the family (See Proposition \ref{prop:q-emb}).} Nonetheless, we may explore at what point this becomes necessary.

Using Theorem \ref{thm:pure-state-conversions-with-LU}, we know the optimal strategy is given by\footnote{We stress that by the correspondence of Schmidt coefficients to probability distributions as discussed at the start of the work, even without Theorem \ref{thm:pure-state-conversions-with-LU}, this would be a legitimate strategy, we simply wouldn't know analytically it was optimal. }
\begin{align*}
 \max_{R \in \cP(d)} F((P \otimes R)^{\downarrow},(Q \otimes R)^{\downarrow}) \ .
\end{align*}
Even in the case $P,Q,R \in \cP(2)$ this technically can't be solved using gradient methods as one has to sort the $p(1-r)$ and $(1-p)r$ terms of $p \otimes r$ and likewise for $q \otimes r$. Nonetheless, it is hopefully clear that $r \in [\min\{p,q\},\max\{p,q\}]$, as it is trying to make the distributions be more similar. Nonetheless, this issue will only grow in difficulty with the dimension and it is unclear how one would prove an ansatz is optimal in general. Therefore, we provide two-qubit examples which characterizes the general insights.
\begin{example}[Resource Gap Between van Dam-Hayden Embezzling State and Optimal Embezzler]\label{example:resource-gap} Consider Bernoulli distributions $P,Q,R$ parameterized by $p = 0.5, q=0.7$ and we leave $r$ unspecified for now. In other words, one of the states is the maximally entangled states and the other is, up to local unitaries, $\sqrt{0.7}\ket{00} + \sqrt{0.3}\ket{11}$. Therefore, depending on which way one runs the transformation, we are considering entanglement dilution or distillation with a catalytic resource. Without the resource, 
$$ F_{LO}(\ket{\psi},\ket{\phi}) = F(P^{\downarrow},Q^{\downarrow}) \approx 0.958  . $$ 
One can verify that the optimal choice of $r^{\star} \approx 0.6$ in this case. For this choice
\begin{align*}
F_{LU}(\ket{\psi}\ket{\zeta}, \ket{\phi}\ket{\zeta}) =& F((P\otimes R^{\star})^{\downarrow},(Q \otimes R^{\star})^{\downarrow}) \\
>& 0.979 \ .
\end{align*}
The first problem is that $0.979$ is not an acceptably high fidelity even by contemporary standards. Nonetheless, note that to get this state via embezzling (and ignoring that embezzling out the initial state introduces error), it would require generating $\ket{\mu(n)}$ where $n > m^{1/(1-0.979)} = 2\cdot 10^{14}$. That is, even to embezzle a two-qubit pure state would require generating an inconceivable amount of entanglement. For this reason, specially engineered \edit{embezzling states} seems a significant improvement up to any error that can be achieved. 

On the other hand, one might note that if we could generate $R$ where $r = 0.55$, then we may as well have just used this state to begin with as 
\begin{align*}
    F(P^{\downarrow},R^{\downarrow}) 
    =& 0.98989 \\
    >& F((P\otimes R^{\star})^{\downarrow}, (Q \otimes R^{\star})) \ .
\end{align*}
From a practical perspective we agree with this critique. Nonetheless, from a basic science perspective, if we are interested in local unitary conversions under \edit{embezzling states/catalysts}, then the above tells us there are better choices in general than \edit{van Dam-Hayden} embezzlement, although van Dam-Hayden embezzling has the special property of being universal and optimal for sufficiently small $\varepsilon$.
\end{example}

\begin{figure}
\centering
\begin{subfigure}[b]{0.95\columnwidth}
    \begin{center}
    \includegraphics[width=\columnwidth]{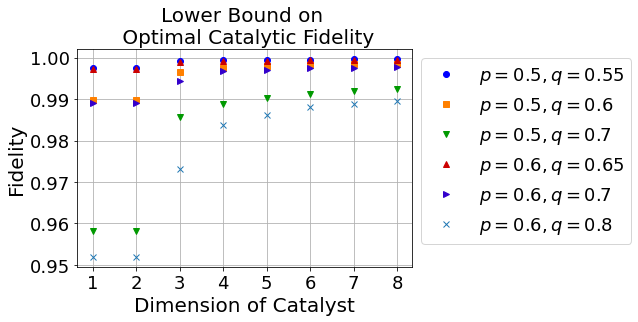}
    \end{center}
   \caption{}
\end{subfigure}
\par\bigskip	
\begin{subfigure}[b]{0.95\columnwidth}
   \begin{center}
    \includegraphics[width=\columnwidth]{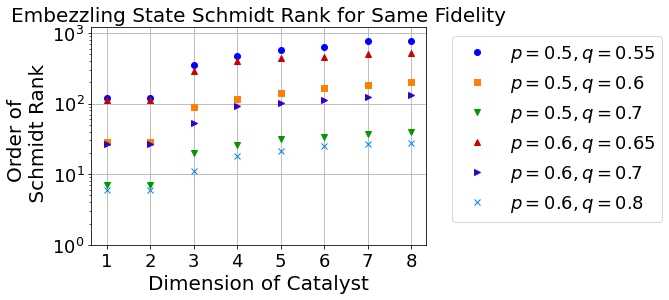}
    \end{center}
   \caption{}
\end{subfigure}
\caption{\small{\edit{Plots pertaining to the dimension scaling of embezzling states motivated by Example \ref{example:resource-gap}. (a) Depicts a lower bound on the achievable fidelity of converting one two-qubit entangled state to another under local unitaries using an \edit{embezzling state} with a given local dimension (equivalently, Schmidt rank) where the two-qubit states are parameterized by  probability distributions $\begin{bmatrix} p & 1- p \end{bmatrix}$ and $\begin{bmatrix} q & 1-q \end{bmatrix}$. These lower bounds were achieved via brute search as described in the main text. (b) We plot the order (i.e. the power of 10) of the Schmidt rank of the van Dam-Hayden embezzling state $\ket{\mu(n)}$ to obtain the same maximum fidelity. This is calculated using $2^{1/(1-\wt{F})}$ following Proposition \ref{prop:q-emb}. See the main text for further information on this calculation. All chosen embezzling states are provided in Appendix \ref{app:catalyst-figure-data} for verification by direct calculation.}}}
\label{fig:catalyst-plots}
\end{figure}

We close this consideration with two final remarks. First, if one picks two states that are more similar to begin with, then the scaling of the embezzling state will be even larger. Second, we have not presented how the fidelity for this example scales as the local dimension of $\ket{\zeta}$ grows. Both the dimension scaling and two states that are more similar are considered in \edit{Fig. \ref{fig:catalyst-plots}. There, in Fig. \ref{fig:catalyst-plots}(a), we find lower bounds on the optimal fidelity achieved for a given size of embezzling state. To do this, we searched over the discretized (ordered) probability simplex where the discretization was over five-thousandths (i.e. $0.005$ intervals in each entry) for dimensions up to $6$. For dimensions $7$ and $8$, due to the time it takes to optimize over such intervals, we would upper bound the search of the $(d+1)^{th}$ entry by the $d^{th}$ entry of the previous dimensions optimizer. This heuristic seems natural as the approximate ordered distributions in each dimension decrease entry-wise almost all the time (See Appendix \ref{app:catalyst-figure-data} which includes all data for the plot). In some low dimensional cases, we are able to verify our solutions are near-optimal by seeing that the optimizer is achieved over a discretization in only hundredths.

To compare to the carefully designed \edit{embezzling states} of Fig.~\ref{fig:catalyst-plots}(a), we consider the sufficient Schmidt rank of the van Dam-Hayden universal embezzling family according to Proposition \ref{prop:q-emb}. This is done via the following calculation. If a fidelity $\wt{F}$ is achieved using the specific embezzling state we found, then we are interested in a van Dam-Hayden universal embezzling state $\ket{\mu(n)}$ with Schmidt rank $n > m^{1/\ve}$ such that $1 - \ve \geq \wt{F}$. Thus, we are interested in $n > 2^{1/(1-\wt{F})}$ where we have used that in our case the state we are embezzling in and out has Schmidt rank $m = 2$. As one would expect, this is a large number so in Fig.~\ref{fig:catalyst-plots}(b), we merely plot the \textit{order} (power of ten) of the Schmidt rank.}

\section{On Extensions of the Theory}\label{sec:on-extensions}
As a final consideration, we discuss the application of our results beyond bipartite pure states. First we remark upon extensions to multipartite pure states. In this case the problem is that in establishing all of the results, we have used that local unitaries can take the Schmidt decomposition of the state to one of a canonical form. However, in the multipartite case, the Schmidt decomposition does not even exist in general \cite{Peres-1995a}. As such this argument immediately breaks down. Furthermore, in the proof of Theorem \ref{thm:pure-state-conversions-with-LU} we used Uhlmann's theorem, which requires partitioning the state into two pieces, one of which is the purification. Therefore, it seems no multipartite extension of this work holds.

Similarly, there are issues with approaching mixed states. One issue is to note that all relationships we have been able to establish have stemmed from the fidelity under local unitaries of pure states. Even in the case where local operations made a pure state no longer pure, we purified operations so that the states were pure. We simply cannot do this if we start with mixed states in both arguments of the fidelity. We also cannot purify the states as by data-processing, any optimization without tracing off the purifying space only gets us a lower bound. Moreover, this lower bound would require establishing results for tripartite systems, which returns to the issues with the multipartite pure state case. Therefore, we believe in effect these are the most general settings where these proof methods will be of use.

\edit{
\section*{Code Availability}
All code used to generate Figures \ref{fig:fidelity-under-LU-plots} and \ref{fig:catalyst-plots} may be found at \href{https://github.com/ChitambarLab/RevisitingPureStateTransformationswithZeroCommunication}{this Git repository} for transparency and scrutiny.

\section*{Acknowledgments}
IG acknowledges support of an Illinois Distinguished Fellowship during the time of this research. This work was supported by NSF Grant No.~2112890.}

\bibliography{References.bib}

\begin{thebibliography}{36}%
\makeatletter
\providecommand \@ifxundefined [1]{%
 \@ifx{#1\undefined}
}%
\providecommand \@ifnum [1]{%
 \ifnum #1\expandafter \@firstoftwo
 \else \expandafter \@secondoftwo
 \fi
}%
\providecommand \@ifx [1]{%
 \ifx #1\expandafter \@firstoftwo
 \else \expandafter \@secondoftwo
 \fi
}%
\providecommand \natexlab [1]{#1}%
\providecommand \enquote  [1]{``#1''}%
\providecommand \bibnamefont  [1]{#1}%
\providecommand \bibfnamefont [1]{#1}%
\providecommand \citenamefont [1]{#1}%
\providecommand \href@noop [0]{\@secondoftwo}%
\providecommand \href [0]{\begingroup \@sanitize@url \@href}%
\providecommand \@href[1]{\@@startlink{#1}\@@href}%
\providecommand \@@href[1]{\endgroup#1\@@endlink}%
\providecommand \@sanitize@url [0]{\catcode `\\12\catcode `\$12\catcode
  `\&12\catcode `\#12\catcode `\^12\catcode `\_12\catcode `\%12\relax}%
\providecommand \@@startlink[1]{}%
\providecommand \@@endlink[0]{}%
\providecommand \url  [0]{\begingroup\@sanitize@url \@url }%
\providecommand \@url [1]{\endgroup\@href {#1}{\urlprefix }}%
\providecommand \urlprefix  [0]{URL }%
\providecommand \Eprint [0]{\href }%
\providecommand \doibase [0]{https://doi.org/}%
\providecommand \selectlanguage [0]{\@gobble}%
\providecommand \bibinfo  [0]{\@secondoftwo}%
\providecommand \bibfield  [0]{\@secondoftwo}%
\providecommand \translation [1]{[#1]}%
\providecommand \BibitemOpen [0]{}%
\providecommand \bibitemStop [0]{}%
\providecommand \bibitemNoStop [0]{.\EOS\space}%
\providecommand \EOS [0]{\spacefactor3000\relax}%
\providecommand \BibitemShut  [1]{\csname bibitem#1\endcsname}%
\let\auto@bib@innerbib\@empty
\bibitem [{\citenamefont {Wheeler}(1989)}]{Wheeler-1989a}%
  \BibitemOpen
  \bibfield  {author} {\bibinfo {author} {\bibfnamefont {J.~A.}\ \bibnamefont
  {Wheeler}},\ }\bibfield  {title} {\bibinfo {title} {Information, physics,
  quantum: The search for links},\ }in\ \href@noop {} {\emph {\bibinfo
  {booktitle} {Proceedings III International Symposium on Foundations of
  Quantum Mechanics}}}\ (\bibinfo {year} {1989})\ pp.\ \bibinfo {pages}
  {354--358}\BibitemShut {NoStop}%
\bibitem [{Zur(1990)}]{Zurek-1990a}%
  \BibitemOpen
  \href@noop {} {\emph {\bibinfo {title} {Complexity, Entropy, and the Physics
  of Information}}},\ Vol.\ \bibinfo {volume} {VIII}\ (\bibinfo  {publisher}
  {Addison-Wesley, The Advanced Book Program},\ \bibinfo {year}
  {1990})\BibitemShut {NoStop}%
\bibitem [{\citenamefont {Landauer}(1991)}]{Landauer-1991a}%
  \BibitemOpen
  \bibfield  {author} {\bibinfo {author} {\bibfnamefont {R.}~\bibnamefont
  {Landauer}},\ }\bibfield  {title} {\bibinfo {title} {Information is
  physical},\ }\href@noop {} {\bibfield  {journal} {\bibinfo  {journal}
  {Physics Today}\ }\textbf {\bibinfo {volume} {44}},\ \bibinfo {pages} {23}
  (\bibinfo {year} {1991})}\BibitemShut {NoStop}%
\bibitem [{\citenamefont {Chitambar}\ and\ \citenamefont
  {Gour}(2019)}]{Chitambar-2019a}%
  \BibitemOpen
  \bibfield  {author} {\bibinfo {author} {\bibfnamefont {E.}~\bibnamefont
  {Chitambar}}\ and\ \bibinfo {author} {\bibfnamefont {G.}~\bibnamefont
  {Gour}},\ }\bibfield  {title} {\bibinfo {title} {Quantum resource theories},\
  }\href@noop {} {\bibfield  {journal} {\bibinfo  {journal} {Reviews of Modern
  Physics}\ }\textbf {\bibinfo {volume} {91}},\ \bibinfo {pages} {025001}
  (\bibinfo {year} {2019})}\BibitemShut {NoStop}%
\bibitem [{\citenamefont {Bell}(1964)}]{Bell-1964}%
  \BibitemOpen
  \bibfield  {author} {\bibinfo {author} {\bibfnamefont {J.~S.}\ \bibnamefont
  {Bell}},\ }\bibfield  {title} {\bibinfo {title} {On the {E}instein {P}odolsky
  {R}osen paradox},\ }\href@noop {} {\bibfield  {journal} {\bibinfo  {journal}
  {Physics Physique Fizika}\ }\textbf {\bibinfo {volume} {1}},\ \bibinfo
  {pages} {195} (\bibinfo {year} {1964})}\BibitemShut {NoStop}%
\bibitem [{\citenamefont {Hayden}\ and\ \citenamefont
  {Winter}(2003)}]{Hayden-2003a}%
  \BibitemOpen
  \bibfield  {author} {\bibinfo {author} {\bibfnamefont {P.}~\bibnamefont
  {Hayden}}\ and\ \bibinfo {author} {\bibfnamefont {A.}~\bibnamefont
  {Winter}},\ }\bibfield  {title} {\bibinfo {title} {Communication cost of
  entanglement transformations},\ }\href@noop {} {\bibfield  {journal}
  {\bibinfo  {journal} {Physical Review A}\ }\textbf {\bibinfo {volume} {67}},\
  \bibinfo {pages} {012326} (\bibinfo {year} {2003})}\BibitemShut {NoStop}%
\bibitem [{\citenamefont {Harrow}\ and\ \citenamefont
  {Lo}(2004)}]{Harrow-2004a}%
  \BibitemOpen
  \bibfield  {author} {\bibinfo {author} {\bibfnamefont {A.}~\bibnamefont
  {Harrow}}\ and\ \bibinfo {author} {\bibfnamefont {H.-K.}\ \bibnamefont
  {Lo}},\ }\bibfield  {title} {\bibinfo {title} {A tight lower bound on the
  classical communication cost of entanglement dilution},\ }\href
  {https://doi.org/10.1109/TIT.2003.822597} {\bibfield  {journal} {\bibinfo
  {journal} {IEEE Transactions on Information Theory}\ }\textbf {\bibinfo
  {volume} {50}},\ \bibinfo {pages} {319} (\bibinfo {year} {2004})}\BibitemShut
  {NoStop}%
\bibitem [{\citenamefont {Wyner}(1975)}]{Wyner-1975a}%
  \BibitemOpen
  \bibfield  {author} {\bibinfo {author} {\bibfnamefont {A.}~\bibnamefont
  {Wyner}},\ }\bibfield  {title} {\bibinfo {title} {The common information of
  two dependent random variables},\ }\href@noop {} {\bibfield  {journal}
  {\bibinfo  {journal} {IEEE Transactions on Information Theory}\ }\textbf
  {\bibinfo {volume} {21}},\ \bibinfo {pages} {163} (\bibinfo {year}
  {1975})}\BibitemShut {NoStop}%
\bibitem [{\citenamefont {Hayashi}(2006)}]{Hayashi-2006a}%
  \BibitemOpen
  \bibfield  {author} {\bibinfo {author} {\bibfnamefont {M.}~\bibnamefont
  {Hayashi}},\ }\href@noop {} {\emph {\bibinfo {title} {Quantum Information: An
  Introduction}}}\ (\bibinfo  {publisher} {Springer},\ \bibinfo {year}
  {2006})\BibitemShut {NoStop}%
\bibitem [{\citenamefont {George}\ \emph {et~al.}(2023)\citenamefont {George},
  \citenamefont {Hsieh},\ and\ \citenamefont {Chitambar}}]{George-2022c}%
  \BibitemOpen
  \bibfield  {author} {\bibinfo {author} {\bibfnamefont {I.}~\bibnamefont
  {George}}, \bibinfo {author} {\bibfnamefont {M.-H.}\ \bibnamefont {Hsieh}},\
  and\ \bibinfo {author} {\bibfnamefont {E.}~\bibnamefont {Chitambar}},\
  }\href@noop {} {\bibinfo {title} {One-shot distributed source simulation: As
  quantum as it can get}} (\bibinfo {year} {2023}),\ \Eprint
  {https://arxiv.org/abs/2301.04301} {arXiv:2301.04301 [quant-ph]} \BibitemShut
  {NoStop}%
\bibitem [{\citenamefont {Schmid}\ \emph {et~al.}(2021)\citenamefont {Schmid},
  \citenamefont {Du}, \citenamefont {Mudassar}, \citenamefont {Coulter-de Wit},
  \citenamefont {Rosset},\ and\ \citenamefont {Hoban}}]{Schmid-2021a}%
  \BibitemOpen
  \bibfield  {author} {\bibinfo {author} {\bibfnamefont {D.}~\bibnamefont
  {Schmid}}, \bibinfo {author} {\bibfnamefont {H.}~\bibnamefont {Du}}, \bibinfo
  {author} {\bibfnamefont {M.}~\bibnamefont {Mudassar}}, \bibinfo {author}
  {\bibfnamefont {G.}~\bibnamefont {Coulter-de Wit}}, \bibinfo {author}
  {\bibfnamefont {D.}~\bibnamefont {Rosset}},\ and\ \bibinfo {author}
  {\bibfnamefont {M.~J.}\ \bibnamefont {Hoban}},\ }\bibfield  {title} {\bibinfo
  {title} {Postquantum common-cause channels: the resource theory of local
  operations and shared entanglement},\ }\href
  {https://doi.org/10.22331/q-2021-03-23-419} {\bibfield  {journal} {\bibinfo
  {journal} {{Quantum}}\ }\textbf {\bibinfo {volume} {5}},\ \bibinfo {pages}
  {419} (\bibinfo {year} {2021})}\BibitemShut {NoStop}%
\bibitem [{\citenamefont {van Dam}\ and\ \citenamefont
  {Hayden}(2003)}]{van-2003a}%
  \BibitemOpen
  \bibfield  {author} {\bibinfo {author} {\bibfnamefont {W.}~\bibnamefont {van
  Dam}}\ and\ \bibinfo {author} {\bibfnamefont {P.}~\bibnamefont {Hayden}},\
  }\bibfield  {title} {\bibinfo {title} {Universal entanglement transformations
  without communication},\ }\href@noop {} {\bibfield  {journal} {\bibinfo
  {journal} {Physical Review A}\ }\textbf {\bibinfo {volume} {67}},\ \bibinfo
  {pages} {060302} (\bibinfo {year} {2003})}\BibitemShut {NoStop}%
\bibitem [{\citenamefont {Bennett}\ \emph {et~al.}(2014)\citenamefont
  {Bennett}, \citenamefont {Devetak}, \citenamefont {Harrow}, \citenamefont
  {Shor},\ and\ \citenamefont {Winter}}]{Bennett-2014a}%
  \BibitemOpen
  \bibfield  {author} {\bibinfo {author} {\bibfnamefont {C.~H.}\ \bibnamefont
  {Bennett}}, \bibinfo {author} {\bibfnamefont {I.}~\bibnamefont {Devetak}},
  \bibinfo {author} {\bibfnamefont {A.~W.}\ \bibnamefont {Harrow}}, \bibinfo
  {author} {\bibfnamefont {P.~W.}\ \bibnamefont {Shor}},\ and\ \bibinfo
  {author} {\bibfnamefont {A.}~\bibnamefont {Winter}},\ }\bibfield  {title}
  {\bibinfo {title} {The quantum reverse shannon theorem and resource tradeoffs
  for simulating quantum channels},\ }\href@noop {} {\bibfield  {journal}
  {\bibinfo  {journal} {IEEE Transactions on Information Theory}\ }\textbf
  {\bibinfo {volume} {60}},\ \bibinfo {pages} {2926} (\bibinfo {year}
  {2014})}\BibitemShut {NoStop}%
\bibitem [{\citenamefont {Anshu}\ \emph {et~al.}(2021)\citenamefont {Anshu},
  \citenamefont {Hadiashar}, \citenamefont {Jain}, \citenamefont {Nayak},\ and\
  \citenamefont {Touchette}}]{Anshu-2021a}%
  \BibitemOpen
  \bibfield  {author} {\bibinfo {author} {\bibfnamefont {A.}~\bibnamefont
  {Anshu}}, \bibinfo {author} {\bibfnamefont {S.~B.}\ \bibnamefont
  {Hadiashar}}, \bibinfo {author} {\bibfnamefont {R.}~\bibnamefont {Jain}},
  \bibinfo {author} {\bibfnamefont {A.}~\bibnamefont {Nayak}},\ and\ \bibinfo
  {author} {\bibfnamefont {D.}~\bibnamefont {Touchette}},\ }\bibfield  {title}
  {\bibinfo {title} {One-shot quantum state redistribution and quantum markov
  chains},\ }in\ \href@noop {} {\emph {\bibinfo {booktitle} {2021 IEEE
  International Symposium on Information Theory (ISIT)}}}\ (\bibinfo
  {organization} {IEEE},\ \bibinfo {year} {2021})\ pp.\ \bibinfo {pages}
  {130--135}\BibitemShut {NoStop}%
\bibitem [{\citenamefont {Leung}\ and\ \citenamefont
  {Wang}(2014)}]{Leung-2014a}%
  \BibitemOpen
  \bibfield  {author} {\bibinfo {author} {\bibfnamefont {D.}~\bibnamefont
  {Leung}}\ and\ \bibinfo {author} {\bibfnamefont {B.}~\bibnamefont {Wang}},\
  }\bibfield  {title} {\bibinfo {title} {Characteristics of universal
  embezzling families},\ }\href {https://doi.org/10.1103/PhysRevA.90.042331}
  {\bibfield  {journal} {\bibinfo  {journal} {Phys. Rev. A}\ }\textbf {\bibinfo
  {volume} {90}},\ \bibinfo {pages} {042331} (\bibinfo {year}
  {2014})}\BibitemShut {NoStop}%
\bibitem [{\citenamefont {Dinur}\ \emph {et~al.}(2015)\citenamefont {Dinur},
  \citenamefont {Steurer},\ and\ \citenamefont {Vidick}}]{Dinur-2015a}%
  \BibitemOpen
  \bibfield  {author} {\bibinfo {author} {\bibfnamefont {I.}~\bibnamefont
  {Dinur}}, \bibinfo {author} {\bibfnamefont {D.}~\bibnamefont {Steurer}},\
  and\ \bibinfo {author} {\bibfnamefont {T.}~\bibnamefont {Vidick}},\
  }\bibfield  {title} {\bibinfo {title} {A parallel repetition theorem for
  entangled projection games},\ }\href@noop {} {\bibfield  {journal} {\bibinfo
  {journal} {Computational Complexity}\ }\textbf {\bibinfo {volume} {24}},\
  \bibinfo {pages} {201} (\bibinfo {year} {2015})}\BibitemShut {NoStop}%
\bibitem [{\citenamefont {Hayden}\ \emph {et~al.}(2004)\citenamefont {Hayden},
  \citenamefont {Jozsa}, \citenamefont {Petz},\ and\ \citenamefont
  {Winter}}]{Hayden-2004a}%
  \BibitemOpen
  \bibfield  {author} {\bibinfo {author} {\bibfnamefont {P.}~\bibnamefont
  {Hayden}}, \bibinfo {author} {\bibfnamefont {R.}~\bibnamefont {Jozsa}},
  \bibinfo {author} {\bibfnamefont {D.}~\bibnamefont {Petz}},\ and\ \bibinfo
  {author} {\bibfnamefont {A.}~\bibnamefont {Winter}},\ }\bibfield  {title}
  {\bibinfo {title} {Structure of states which satisfy strong subadditivity of
  quantum entropy with equality},\ }\href@noop {} {\bibfield  {journal}
  {\bibinfo  {journal} {Communications in mathematical physics}\ }\textbf
  {\bibinfo {volume} {246}},\ \bibinfo {pages} {359} (\bibinfo {year}
  {2004})}\BibitemShut {NoStop}%
\bibitem [{\citenamefont {Gour}(2024)}]{Gour-2024a}%
  \BibitemOpen
  \bibfield  {author} {\bibinfo {author} {\bibfnamefont {G.}~\bibnamefont
  {Gour}},\ }\href@noop {} {\bibinfo {title} {Resources of the quantum world}}
  (\bibinfo {year} {2024}),\ \Eprint {https://arxiv.org/abs/2402.05474}
  {arXiv:2402.05474 [quant-ph]} \BibitemShut {NoStop}%
\bibitem [{\citenamefont {Buscemi}(2012)}]{Buscemi-2012a}%
  \BibitemOpen
  \bibfield  {author} {\bibinfo {author} {\bibfnamefont {F.}~\bibnamefont
  {Buscemi}},\ }\bibfield  {title} {\bibinfo {title} {All entangled quantum
  states are nonlocal},\ }\href@noop {} {\bibfield  {journal} {\bibinfo
  {journal} {Physical review letters}\ }\textbf {\bibinfo {volume} {108}},\
  \bibinfo {pages} {200401} (\bibinfo {year} {2012})}\BibitemShut {NoStop}%
\bibitem [{\citenamefont {Aubrun}\ and\ \citenamefont
  {Nechita}(2008)}]{Aubrun-2008a}%
  \BibitemOpen
  \bibfield  {author} {\bibinfo {author} {\bibfnamefont {G.}~\bibnamefont
  {Aubrun}}\ and\ \bibinfo {author} {\bibfnamefont {I.}~\bibnamefont
  {Nechita}},\ }\bibfield  {title} {\bibinfo {title} {Catalytic majorization
  and norms},\ }\href@noop {} {\bibfield  {journal} {\bibinfo  {journal}
  {Communications in Mathematical Physics}\ }\textbf {\bibinfo {volume}
  {278}},\ \bibinfo {pages} {133} (\bibinfo {year} {2008})}\BibitemShut
  {NoStop}%
\bibitem [{\citenamefont {Leung}\ \emph {et~al.}(2013)\citenamefont {Leung},
  \citenamefont {Toner},\ and\ \citenamefont {Watrous}}]{Leung-2011a}%
  \BibitemOpen
  \bibfield  {author} {\bibinfo {author} {\bibfnamefont {D.}~\bibnamefont
  {Leung}}, \bibinfo {author} {\bibfnamefont {B.}~\bibnamefont {Toner}},\ and\
  \bibinfo {author} {\bibfnamefont {J.}~\bibnamefont {Watrous}},\ }\bibfield
  {title} {\bibinfo {title} {Coherent state exchange in multi-prover quantum
  interactive proof systems},\ }\href@noop {} {\bibfield  {journal} {\bibinfo
  {journal} {Chicago Journal of Theoretical Computer Science}\ }\textbf
  {\bibinfo {volume} {11}},\ \bibinfo {pages} {1} (\bibinfo {year}
  {2013})}\BibitemShut {NoStop}%
\bibitem [{\citenamefont {Brandao}\ \emph {et~al.}(2015)\citenamefont
  {Brandao}, \citenamefont {Horodecki}, \citenamefont {Ng}, \citenamefont
  {Oppenheim},\ and\ \citenamefont {Wehner}}]{Brandao-2015a}%
  \BibitemOpen
  \bibfield  {author} {\bibinfo {author} {\bibfnamefont {F.}~\bibnamefont
  {Brandao}}, \bibinfo {author} {\bibfnamefont {M.}~\bibnamefont {Horodecki}},
  \bibinfo {author} {\bibfnamefont {N.}~\bibnamefont {Ng}}, \bibinfo {author}
  {\bibfnamefont {J.}~\bibnamefont {Oppenheim}},\ and\ \bibinfo {author}
  {\bibfnamefont {S.}~\bibnamefont {Wehner}},\ }\bibfield  {title} {\bibinfo
  {title} {The second laws of quantum thermodynamics},\ }\href@noop {}
  {\bibfield  {journal} {\bibinfo  {journal} {Proceedings of the National
  Academy of Sciences}\ }\textbf {\bibinfo {volume} {112}},\ \bibinfo {pages}
  {3275} (\bibinfo {year} {2015})}\BibitemShut {NoStop}%
\bibitem [{\citenamefont {Ng}\ \emph {et~al.}(2015)\citenamefont {Ng},
  \citenamefont {Man{\v{c}}inska}, \citenamefont {Cirstoiu}, \citenamefont
  {Eisert},\ and\ \citenamefont {Wehner}}]{Ng-2015a}%
  \BibitemOpen
  \bibfield  {author} {\bibinfo {author} {\bibfnamefont {N.~H.~Y.}\
  \bibnamefont {Ng}}, \bibinfo {author} {\bibfnamefont {L.}~\bibnamefont
  {Man{\v{c}}inska}}, \bibinfo {author} {\bibfnamefont {C.}~\bibnamefont
  {Cirstoiu}}, \bibinfo {author} {\bibfnamefont {J.}~\bibnamefont {Eisert}},\
  and\ \bibinfo {author} {\bibfnamefont {S.}~\bibnamefont {Wehner}},\
  }\bibfield  {title} {\bibinfo {title} {Limits to catalysis in quantum
  thermodynamics},\ }\href@noop {} {\bibfield  {journal} {\bibinfo  {journal}
  {New Journal of Physics}\ }\textbf {\bibinfo {volume} {17}},\ \bibinfo
  {pages} {085004} (\bibinfo {year} {2015})}\BibitemShut {NoStop}%
\bibitem [{\citenamefont {Datta}\ \emph {et~al.}(2023)\citenamefont {Datta},
  \citenamefont {Kondra}, \citenamefont {Miller},\ and\ \citenamefont
  {Streltsov}}]{Datta-2023a}%
  \BibitemOpen
  \bibfield  {author} {\bibinfo {author} {\bibfnamefont {C.}~\bibnamefont
  {Datta}}, \bibinfo {author} {\bibfnamefont {T.~V.}\ \bibnamefont {Kondra}},
  \bibinfo {author} {\bibfnamefont {M.}~\bibnamefont {Miller}},\ and\ \bibinfo
  {author} {\bibfnamefont {A.}~\bibnamefont {Streltsov}},\ }\bibfield  {title}
  {\bibinfo {title} {Catalysis of entanglement and other quantum resources},\
  }\href@noop {} {\bibfield  {journal} {\bibinfo  {journal} {Reports on
  Progress in Physics}\ } (\bibinfo {year} {2023})}\BibitemShut {NoStop}%
\bibitem [{\citenamefont {Lipka-Bartosik}\ \emph {et~al.}(2023)\citenamefont
  {Lipka-Bartosik}, \citenamefont {Wilming},\ and\ \citenamefont
  {Ng}}]{Lipka-Bartosik-2023a}%
  \BibitemOpen
  \bibfield  {author} {\bibinfo {author} {\bibfnamefont {P.}~\bibnamefont
  {Lipka-Bartosik}}, \bibinfo {author} {\bibfnamefont {H.}~\bibnamefont
  {Wilming}},\ and\ \bibinfo {author} {\bibfnamefont {N.~H.~Y.}\ \bibnamefont
  {Ng}},\ }\href@noop {} {\bibinfo {title} {Catalysis in quantum information
  theory}} (\bibinfo {year} {2023}),\ \Eprint
  {https://arxiv.org/abs/2306.00798} {arXiv:2306.00798 [quant-ph]} \BibitemShut
  {NoStop}%
\bibitem [{\citenamefont {Wilde}(2013)}]{Wilde-2011a}%
  \BibitemOpen
  \bibfield  {author} {\bibinfo {author} {\bibfnamefont {M.~M.}\ \bibnamefont
  {Wilde}},\ }\bibfield  {title} {\bibinfo {title} {Quantum information
  theory},\ }\href@noop {} {\  (\bibinfo {year} {2013})}\BibitemShut {NoStop}%
\bibitem [{\citenamefont {Watrous}(2018)}]{Watrous-Book}%
  \BibitemOpen
  \bibfield  {author} {\bibinfo {author} {\bibfnamefont {J.}~\bibnamefont
  {Watrous}},\ }\href {https://doi.org/10.1017/9781316848142} {\emph {\bibinfo
  {title} {The Theory of Quantum Information}}}\ (\bibinfo  {publisher}
  {Cambridge University Press},\ \bibinfo {year} {2018})\BibitemShut {NoStop}%
\bibitem [{\citenamefont {de~Vicente}\ and\ \citenamefont
  {Huber}(2011)}]{de-2011a}%
  \BibitemOpen
  \bibfield  {author} {\bibinfo {author} {\bibfnamefont {J.~I.}\ \bibnamefont
  {de~Vicente}}\ and\ \bibinfo {author} {\bibfnamefont {M.}~\bibnamefont
  {Huber}},\ }\bibfield  {title} {\bibinfo {title} {Multipartite entanglement
  detection from correlation tensors},\ }\href@noop {} {\bibfield  {journal}
  {\bibinfo  {journal} {Physical Review A}\ }\textbf {\bibinfo {volume} {84}},\
  \bibinfo {pages} {062306} (\bibinfo {year} {2011})}\BibitemShut {NoStop}%
\bibitem [{\citenamefont {Mudholkar}\ and\ \citenamefont
  {Freimer}(1985)}]{Mudholkar-1985a}%
  \BibitemOpen
  \bibfield  {author} {\bibinfo {author} {\bibfnamefont {G.~S.}\ \bibnamefont
  {Mudholkar}}\ and\ \bibinfo {author} {\bibfnamefont {M.}~\bibnamefont
  {Freimer}},\ }\bibfield  {title} {\bibinfo {title} {A structure theorem for
  the polars of unitarily invariant norms},\ }\href@noop {} {\bibfield
  {journal} {\bibinfo  {journal} {Proceedings of the American Mathematical
  Society}\ }\textbf {\bibinfo {volume} {95}},\ \bibinfo {pages} {331}
  (\bibinfo {year} {1985})}\BibitemShut {NoStop}%
\bibitem [{\citenamefont {Johnston}(2012)}]{Johnston-2012a}%
  \BibitemOpen
  \bibfield  {author} {\bibinfo {author} {\bibfnamefont {N.}~\bibnamefont
  {Johnston}},\ }\emph {\bibinfo {title} {Norms and Cones in the Theory of
  Quantum Entanglement}},\ \href {https://doi.org/10.48550/ARXIV.1207.1479}
  {Ph.D. thesis},\ \bibinfo  {school} {University of Guelph} (\bibinfo {year}
  {2012})\BibitemShut {NoStop}%
\bibitem [{\citenamefont {Berta}\ \emph {et~al.}(2010)\citenamefont {Berta},
  \citenamefont {Christandl}, \citenamefont {Colbeck}, \citenamefont {Renes},\
  and\ \citenamefont {Renner}}]{Berta-2010a}%
  \BibitemOpen
  \bibfield  {author} {\bibinfo {author} {\bibfnamefont {M.}~\bibnamefont
  {Berta}}, \bibinfo {author} {\bibfnamefont {M.}~\bibnamefont {Christandl}},
  \bibinfo {author} {\bibfnamefont {R.}~\bibnamefont {Colbeck}}, \bibinfo
  {author} {\bibfnamefont {J.~M.}\ \bibnamefont {Renes}},\ and\ \bibinfo
  {author} {\bibfnamefont {R.}~\bibnamefont {Renner}},\ }\bibfield  {title}
  {\bibinfo {title} {The uncertainty principle in the presence of quantum
  memory},\ }\href@noop {} {\bibfield  {journal} {\bibinfo  {journal} {Nature
  Physics}\ }\textbf {\bibinfo {volume} {6}},\ \bibinfo {pages} {659} (\bibinfo
  {year} {2010})}\BibitemShut {NoStop}%
\bibitem [{\citenamefont {Horn}\ and\ \citenamefont
  {Johnson}(2013)}]{Horn-2012a}%
  \BibitemOpen
  \bibfield  {author} {\bibinfo {author} {\bibfnamefont {R.~A.}\ \bibnamefont
  {Horn}}\ and\ \bibinfo {author} {\bibfnamefont {C.~R.}\ \bibnamefont
  {Johnson}},\ }\href@noop {} {\emph {\bibinfo {title} {Matrix analysis}}},\
  \bibinfo {edition} {2nd}\ ed.\ (\bibinfo  {publisher} {Cambridge university
  press},\ \bibinfo {year} {2013})\BibitemShut {NoStop}%
\bibitem [{\citenamefont {Bennett}\ \emph {et~al.}(1996)\citenamefont
  {Bennett}, \citenamefont {Bernstein}, \citenamefont {Popescu},\ and\
  \citenamefont {Schumacher}}]{Bennett-1996a}%
  \BibitemOpen
  \bibfield  {author} {\bibinfo {author} {\bibfnamefont {C.~H.}\ \bibnamefont
  {Bennett}}, \bibinfo {author} {\bibfnamefont {H.~J.}\ \bibnamefont
  {Bernstein}}, \bibinfo {author} {\bibfnamefont {S.}~\bibnamefont {Popescu}},\
  and\ \bibinfo {author} {\bibfnamefont {B.}~\bibnamefont {Schumacher}},\
  }\bibfield  {title} {\bibinfo {title} {Concentrating partial entanglement by
  local operations},\ }\href@noop {} {\bibfield  {journal} {\bibinfo  {journal}
  {Physical Review A}\ }\textbf {\bibinfo {volume} {53}},\ \bibinfo {pages}
  {2046} (\bibinfo {year} {1996})}\BibitemShut {NoStop}%
\bibitem [{\citenamefont {Cover}\ and\ \citenamefont
  {Thomas}(2006)}]{Cover-2006}%
  \BibitemOpen
  \bibfield  {author} {\bibinfo {author} {\bibfnamefont {T.~M.}\ \bibnamefont
  {Cover}}\ and\ \bibinfo {author} {\bibfnamefont {J.~A.}\ \bibnamefont
  {Thomas}},\ }\href@noop {} {\emph {\bibinfo {title} {Elements of Information
  Theory}}}\ (\bibinfo  {publisher} {John Wiley \& Sons, Inc.},\ \bibinfo
  {year} {2006})\BibitemShut {NoStop}%
\bibitem [{\citenamefont {Peres}(1995)}]{Peres-1995a}%
  \BibitemOpen
  \bibfield  {author} {\bibinfo {author} {\bibfnamefont {A.}~\bibnamefont
  {Peres}},\ }\bibfield  {title} {\bibinfo {title} {Higher order schmidt
  decompositions},\ }\href@noop {} {\bibfield  {journal} {\bibinfo  {journal}
  {arXiv preprint quant-ph/9504006}\ } (\bibinfo {year} {1995})}\BibitemShut
  {NoStop}%
\bibitem [{\citenamefont {Yu}\ and\ \citenamefont {Tan}(2022)}]{Yu-2022a}%
  \BibitemOpen
  \bibfield  {author} {\bibinfo {author} {\bibfnamefont {L.}~\bibnamefont
  {Yu}}\ and\ \bibinfo {author} {\bibfnamefont {V.~Y.~F.}\ \bibnamefont
  {Tan}},\ }\bibfield  {title} {\bibinfo {title} {Common information, noise
  stability, and their extensions},\ }\href
  {https://doi.org/10.1561/0100000122} {\bibfield  {journal} {\bibinfo
  {journal} {Foundations and Trends® in Communications and Information
  Theory}\ }\textbf {\bibinfo {volume} {19}},\ \bibinfo {pages} {107} (\bibinfo
  {year} {2022})}\BibitemShut {NoStop}%
\end{thebibliography}%
\appendix

\section{Randomness Embezzling Proof and Discussion on Locality}\label{app:rand-embez-proof}
In this section we provide the proof of Proposition \ref{prop:classical-embezzling} and then briefly discuss how it differs from quantum embezzlement.
\begin{proof}
The proof is largely the same as for embezzlement of quantum states \cite{van-2003a}. Let $P = \sum_{i} p(i) \dyad{i}$. Define $W_{n}$ as $R_{n} \otimes P$ except with probabilities in decreasing order. Note 
$$ R_{n} \otimes P = \frac{1}{H_{n}} \sum_{i,j} \frac{p(i)}{j} \dyad{i} \otimes \dyad{j} \ , $$
so there exists a relabeling on $\{(i,j)\}$ that will take this to $W_{n}$. In particular, letting $f : [m] \times [n] \to [m \cdot n]$ be a bijection, we have $\ket{i}\ket{j} \to \ket{f(i,j)} \equiv \ket{i'}\ket{j'}$ such that $\left \{z_{f(i,j)} := \frac{p(i)}{jH_{n}} \right \}_{(i,j)}$ satisfy $z_{k} \geq z_{k+1}$ for all $k \in [m \cdot n]$. Therefore it suffices to approximate $W_{n}$, which means we want to bound the overlap of this with $R_{n} \otimes P$. 

For fixed $t$ and $i$, we let $$N^{t}_{i} := \left|\left\{(i,j) : \frac{p(i)}{jH_{n}} > \frac{1}{tH_{n}} \right \} \right|  \ . $$
The inequality may be manipulated to imply $1 \leq j < p(i)t$. It follows that $N^{t}_{i} < p(i)t$. From this we obtain $\sum_{i=1}^{m}N^{t}_{i} < \sum_{i=1}^{m} p(i)t < t$, where we have used $\sum_{i} p(i) = 1$. As $z_{1} \geq z_{2} \geq ...$, it follows $z_{j} \leq \frac{1}{jH_{n}}$ for all $1 \geq j \geq n$. We may restate this as for $1 \leq j \leq n$, there are at most $t'-1$ pairs $(i,j)$ such that $p(i)/(jH_{n}) > 1/(t'H_{n})$. Recalling $z_{1} \geq z_{2} \geq ...$, this means that $z_{1} < 1/H_{n}$ and that there is at most one pair $(i,j)$ pair such that $p(i)/(jH_{n})<1/(2H_{n})$, which, since $z_{1} \geq z_{2}$, means if such a pair exists, it is $z_{1}$. By applying this argument in effect recursively, we see that for $t'$, there are at most $t'-1$ $(i,j)$ pairs such that $p(i)/(jH_{n}) > 1/(t' H_{n})$ and since $z_{k} \geq z_{k+1}$, if all of these pairs exist, then it must be $z_{1},...,z_{t'-1}$. Therefore, $z_{j} \leq 1/(jH_{n})$ for all $1 \leq j \leq n$. We can now use this to bound the fidelity.
\begin{align*}
    F(R_{n} \otimes \dyad{0} , W_{n}) = & \left(\sum_{j=1}^{n} \sqrt{\frac{z_{j}}{jH_{n}}} \right)^{2}  \\
    \geq & \left( \sum_{j=1}^{n} \sqrt{z_{j}} \right)^{2}
    \geq \sum_{j = 1}^{n} z_{j},
\end{align*}
where in the equality we have used the definition of fidelity, in the second we used our established inequality, and in the third we have used $\sqrt{x} + \sqrt{y} \geq \sqrt{x+y}$ for $x,y \geq 0$ to pull the square root out around the sum and cancel with the square.

Now we want to lower bound this sum, which requires managing the $z_{j}$ terms. We consider $T_{n} = R_{n} \otimes \pi_{m}$ with probabilities $t(j)$ where $\pi_{m} := \frac{1}{m} \sum_{i = 1}^{m} \dyad{i}$. Now note that $z_{k} \geq t_{k}$ for all $k \in [m\cdot n]$, and this is independent of what the distribution $P$ is. We can then bound the relevant sum by the sum for $T_{n}$. It follows
\begin{align*}
    \sum_{j=1}^{n} t_{j} = \sum_{j=1}^{\lfloor n/m \rfloor} \sum_{i=1}^{m} \frac{1}{jH_{n}m} =&  \sum_{j=1}^{\lfloor n/m \rfloor} \frac{1}{jH_{n}} \\
    =& \frac{H_{\lfloor n /m \rfloor}}{H_{m}} \\
    \geq& \frac{\ln(n/m)}{\ln(n)} = 1 - \frac{\log(m)}{\log(n)} \ ,
\end{align*}
where the second inequality is using $H_{n} \geq \ln(n)$ and the final form is converting from $\ln$ to $\log$ in both the numerator and denominator so it cancels. Finally, leting $1 - \log(m)/\log(n) > 1 - \ve$ will result in $n > m^{1/\ve}$, which completes the proof.
\end{proof}

With the proof established, we expand upon the distinction between the entangled and classical distribution cases of embezzlement in terms of locality briefly mentioned in the main text. In the classical case, one party embezzles a distribution locally by themselves, whereas in the entangled case two parties act locally on a non-local distribution. Mathematically, this simply follows from the fact the $\opvec(\cdot)$ map and its inverse converts between bipartite states and a probability distribution. However, it is also physically interesting that these are the two cases that align as it is clear other variations are either classically or quantumly impossible as we now explain.

The first reasonable variation would be if there is a non-local classical case where two parties try and construct some joint distribution $p_{XY}$ using \edit{classical embezzler} $r_{X'Y'}$. It is easy to see that they cannot in general satisfy the decoupling condition that is satisfied in quantum embezzlement, i.e.\ they cannot satisfy $p_{XY} \otimes r_{X'Y'}$ in this setting. This is because without loss of generality the state will be of the form
\begin{align*} 
q_{XYX'Y'} =& \sum_{x,x',y,y'} q(x|x')q'(y|y')r(x,y) \\
& \hspace{2cm} \cdot \dyad{x,y,x',y'} \ .
\end{align*}
This form means that $X$ will be correlated to $X'$ and $Y$ to $Y'$ unless $q_{XY}$ may be generated non-locally without a seed state to correlate the two which means they are (up to the allowed error) independent, i.e. $q_{XY} \approx_{\ve} q_{X} \otimes q_{Y}$. In this sense, there cannot be a classical non-local equivalent of quantum embezzlement.

On the other hand, if one does not require the decoupling, then this is a task that is possible in the classical setting and is known as distributed source simulation, where the question is the minimal needed \textit{shared randomness} as the \textit{seed state} to generate the target state up to an (arbitrary) error \cite{Yu-2022a}. This was determined asymptotically in the classical case by Wyner \cite{Wyner-1975a}, extended to separable states by Hayashi \cite{Hayashi-2006a}, and recently generalized to the one-shot setting for separable states in \cite{George-2022c}. However, as in this setting variation there is no communication between the acting parties and the \edit{embezzler} acts as the seed state, it follows from Proposition \ref{prop:Hayden-Winter} that distributed source simulation cannot admit an entangled state equivalent. For these reasons, not only does the $\opvec$ bijection specify the correspondence of embezzlement in the classical and quantum setting, but deviating from it makes either a quantum or classical version impossible.

\section{Semidefinite Program Relaxation of Max Fidelity of Pure State Transformation Under LOSR}\label{app:SDP-relaxation}
In this section we prove Theorem \ref{thm:sdp-upper-bound}. We begin by establishing \eqref{eq:SDP-upper-bound} is true.
\begin{lemma}
Consider target state $\ket{\psi}$ and seed state $\ket{\phi}$. Let $\mrm{SR}(\psi) = d$ and $\mrm{SR}(\phi) = d'$. Define $A = \mbb{C}^{d}$, $B = \mbb{C}^{d \cdot d'}$. Then,
\begin{equation*}
    \begin{aligned}
    F_{\mrm{LOSR}}(\ket{\psi},\ket{\phi}) \leq \max & \; F(R,Q^{\downarrow}_{\text{embed}}) \\
    \mrm{s.t.} & \; \Tr_{B}[R] = P^{\downarrow} \\
    & \; R \in \cP^{\downarrow}(d^{2}\cdot d') \ ,
    \end{aligned}
\end{equation*}
where $P$ and $Q$ are the distributions defined by $\ket{\psi}$ and $\ket{\phi}$'s Schmidt coefficients respectively.
\end{lemma}
\begin{proof}
The above seems intuitively true from Theorem \ref{thm:pure-state-conversion-via-LO} as we have just relaxed the tensor product structure with the partial trace constraint. The technical issue is the ordering operation $\cdot^{\downarrow}$ is defined in terms of a permutation of a fixed basis, so we need to make sure this works with the partial trace.

Note the feasible set, the set we can optimizer over, in Theorem \ref{thm:pure-state-conversion-via-LO} is $S_{1}(P) := \{(P\otimes P')^{\downarrow}: P' \in \cP(\Sigma)\}$. Now note this is the same as the set 
$$ S_{2}(P) := \{ (P^{\downarrow} \otimes {P'}^{\downarrow})^{\downarrow}: P' \in \cP(\Sigma) \} \ , $$
because the ordering applied to the tensor product will result in the same thing regardless of whether or not $P,P'$ were ordered. Therefore, we can focus on $P^{\downarrow} \otimes {P'}^{\downarrow}$ to make the explanation clearer. 

In general, in terms of vectors,
$$ (p^{\downarrow} \otimes {p'}^{\downarrow})^{\downarrow} = \begin{pmatrix} p^{\downarrow}(1){p'}^{\downarrow} \\ p^{\downarrow}(2){p'}^{\downarrow} \\ \vdots \\ p^{\downarrow}(d){p'}^{\downarrow} \end{pmatrix} \ , $$
where $p(i) \geq p(i+k)$ for $k \geq 0$. Formally, we also have
$$ p^{\downarrow}(i){p'}^{\downarrow}(1) \geq p^{\downarrow}(i+k){p'}^{\downarrow}(j) $$
for all $i \in [d]$, $k \in \{0,...,d-i\}$, and 
$j \in \Sigma$. In particular what this means is that without loss of generality for any $i \in [d]$, $p^{\downarrow}(i){p'}^{\downarrow}(1)$ appears before any element that is \textit{not} of the form $p^{\downarrow}(i-\ell){p'}^{\downarrow}(j)$ for some $0 < \ell \leq i-1$. It follows that under the ordering of $ (p^{\downarrow} \otimes {p'}^{\downarrow})^{\downarrow}$, when the partial trace marginalizes to the $A$ space, the induced ordering on the local space will be the ordering based on $p^{\downarrow}$. Formally, this can be expressed as
\begin{align*}
    & \Tr_{\mbb{C}^{|\Sigma|}}[(P^{\downarrow} \otimes {P'}^{\downarrow})^{\downarrow}] \\
    = & \sum_{j \in \Sigma} \mbb{1}_{A} \otimes \bra{j} (P^{\downarrow} \otimes {P'}^{\downarrow})^{\downarrow} \ket{j} \\
    = & \sum_{i \in [d]} p^{\downarrow}(i)\dyad{i} \ ,
\end{align*}
where the first equality is a representation of the partial trace and the second is using the property noted of the ordering on the joint ordered distribution.

Thus, if $X \in S_{2}(P)$, $\Tr_{C^{|\Sigma|}}(X) = P^{\downarrow}$ and $X \in \cP^{\downarrow}(d \cdot |\Sigma|)$. Noting that $|\Sigma| = d \cdot d'$, this is the feasible set we have defined in the proposition. This completes the proof.
\end{proof}

The remaining point is to prove this is the semidefinite program given in \eqref{eq:Reals-SDP-representation}. There is much to the theory of semidefinite programs for quantum information \cite{Watrous-Book}, but for our purposes all we will need is the following definition.
\begin{definition}
A semidefinite program may be expressed as 
\begin{equation*}
    \begin{aligned}
        \max & \Tr(AX) \\
        \text{s.t.} & \; \Phi(X) = B \\
        & X_{\mbb{C}^{d}} \preceq 0 \ ,
    \end{aligned}
\end{equation*}
where $\Phi \in \Trans(\mbb{C}^{d},\mbb{C}^{d'})$ is a Hermitian-preserving map, $A \in \Herm(\mbb{C}^{d})$, $B \in \Herm(\mbb{C}^{d'})$, and $\Herm(\cdot)$ is the space of Hermitian operators on a given Hilbert space.
\end{definition}

The fidelity is known to be a semidefinite program \cite{Watrous-Book}, so we are really just verifying all of our constraints work and that we can write the SDP simply by making use of that.

\begin{lemma}\label{lem:SDP-representation}
The optimization program in the previous lemma, may be expressed as the following semidefinite program over the reals.
\begin{equation*}
    \begin{aligned}
        \max & \sum_{i \in [d^{2} \cdot d']} x(i) \\
        \text{s.t.} & \; \begin{pmatrix} \text{diag}(r) & \text{diag}(x) \\ \text{diag}(x) & \text{diag}(q^{\downarrow}_{\text{embed}}) \end{pmatrix} \succeq 0 \\
        & \Tr_{B}[\text{diag}(r)] = P^{\downarrow} \\
        & r \in \cP^{\downarrow}([d^{2} \cdot d]) \\ 
        & x \in \mbb{R}^{d^{2}\cdot d'} \ ,
    \end{aligned}
\end{equation*}
where $d,d'$ are defined in the previous lemma.
\end{lemma}
\begin{proof}
We begin by expressing the objective function of the previous lemma, which is in terms of fidelity, using the primal problem for the SDP for fidelity from \cite[Theorem 3.17]{Watrous-Book}:
\begin{equation*}
    \begin{aligned}
        \max & \frac{1}{2}\left[\Tr(X) + \Tr(X^{\dagger}) \right] \\
        & \; \begin{pmatrix} R & X \\ X^{\dagger} & Q^{\downarrow}_{\text{embed}} \end{pmatrix} \geq 0 \\ 
        & X \in \Lin(\mbb{C}^{[d^{2} \cdot d']}) \ .
    \end{aligned}
\end{equation*}
Now our goal is to reduce $X$ to the diagonal of a real vector.

Note that $R,Q^{\downarrow}_{\text{embed}}$ are always invariant under pinching onto the computational basis of $\mbb{C}^{[d^{2} \cdot d']}$, which we can denote $\Delta$. Note that this pinching is a CPTP, so by the CP property,
\begin{align*}
    & (\mrm{id}_{\mbb{C}^{2}} \otimes \Delta)\begin{pmatrix} R & X \\ X^{\dagger} & Q^{\downarrow}_{\text{embed}} \end{pmatrix} \\ =& \begin{pmatrix} R & \Delta(X) \\ \Delta(X^{\dagger}) & Q^{\downarrow}_{\text{embed}} \end{pmatrix} \geq 0 \ .
\end{align*}
It also then follows as a positive semidefinite operator is always Hermitian that
\begin{align*}
    & \begin{pmatrix} R & \Delta(X^{\dagger}) \\ \Delta(X) & Q^{\downarrow}_{\text{embed}} \end{pmatrix} \geq 0 \ .
\end{align*}
Thus by taking these two cases and averaging them, we have that
\begin{align*}
    \begin{pmatrix} R & \frac{1}{2}\left(\Delta(X + X^{\dagger})\right) \\ \frac{1}{2}\left(\Delta(X + X^{\dagger})\right) & Q^{\downarrow}_{\text{embed}} \end{pmatrix} \geq 0 \ .
\end{align*}
Define $\overline{X} := \frac{1}{2}\left(\Delta(X + X^{\dagger})\right)$. Then note 
\begin{align*}
& \frac{1}{2}\left(\Tr(X) + \Tr(X^{\dagger})\right) \\
=& \frac{1}{2}\left(\Tr(\Delta(X)) + \Tr(\Delta(X^{\dagger}))\right) \\
=& \frac{1}{2}\left(\Tr(\overline{X}) + \Tr(\overline{X}^{\dagger})\right) = \Tr(\overline{X}) \ ,
\end{align*}
where the first equality is because the pinching is trace preserving, the second is by definition of $\overline{X}$, as is the final equality. Thus, for any $X$ that satisfies the positivity constraint, we could replace it with $\overline{X}$ without loss of generality as we are considering a maximization. Finally, note that $\overline{X}$ is a real diagonal matrix by the pinching along with the fact $a + a^{\ast} = 2\Re{a}$. Thus $\overline{X} = \text{diag}(x)$ for some $x \in \mbb{R}^{d^{2} \cdot d'}$. Combining all these points and using $\Tr(\overline{X}) = \sum_{i \in [d^{2} \cdot d']} x(i)$, we have reduced to considering 
\begin{equation*}
    \begin{aligned}
        \max & \sum_{i \in [d^{2} \cdot d']} x(i) \\
        & \; \begin{pmatrix} \text{diag}(r) & \text{diag}(x) \\ \text{diag}(x) & \text{diag}(q^{\downarrow}_{\text{embed}}) \end{pmatrix} \geq 0 \\ 
        & x \in \mbb{R}^{d^{2} \cdot d'} \ .
    \end{aligned}
\end{equation*}
This argument works for any choice of diagonal $r$, so this is the major reduction.

What remains is to prove all the constraints are Hermitian maps. One can write the constraints for $r \in \cP^{\downarrow}$ as $r(i) \geq r(i+1)$ for all $i$, which are semidefinite constraints and can be written as Hermitian preserving maps on the variables $r,x$. $\text{diag}$ is a Hermitian preserving map as is the partial trace, so $\Tr_{C}[\text{diag}(r)]$ is a Hermitian preserving map. Likewise is the block matrix mapping if one allows for the complex conjugate in the lower left block, but noting $\text{diag}(x)^{\dagger} = \text{diag}(x)$, we can leave it as written. Thus all the maps are Hermitian-preserving.

The conversion to actual standard form we then omit as it provides no insight. This completes the proof.
\end{proof}
The above two proofs establish Theorem \ref{thm:sdp-upper-bound}.

\section{Data for Embezzling State Figure}
\label{app:catalyst-figure-data}
\edit{In this appendix, we provide all the embezzling distributions used to generate Fig.\ \ref{fig:catalyst-plots}.
\begin{table}[H]
\centering
\begin{tabular}{r|l}
     $d$ & Lower Bound Embezzling Distribution $r$   \\ \hline
     1 &  n/a \\
     2 & [0.5,0.5] \\
     3 & [0.3,0.33,0.37] \\
     4 & [0.215,0.235,0.26,0.29] \\
     5 & [0.16,0.18,0.2,0.22,0.24] \\
     6 & [0.13,0.145,0.155,0.17,0.19,0.21] \\
     7 & [0.105,0.115,0.125,0.14,0.155,0.17,0.19] \\
     8 & [0.085,0.095,0.105,0.115,0.13,0.115,0.13,0.145,0.155,0.17] 
\end{tabular}
\caption{Data for the near optimal catalyst results presented in Fig.\ \ref{fig:catalyst-plots} in the case $p=0.5$ and $q=0.55$. $d$ stands for the dimension of the distribution. Each distribution $r$ was found via numerical search as described in the main text.}
\end{table}

\begin{table}[H]
\centering
\begin{tabular}{r|l}
     $d$ & Lower Bound Embezzling Distribution $r$   \\ \hline
     1 &  n/a \\
     2 & [0.5,0.5] \\
     3 & [0.27,0.33,0.4]	\\
     4 & [0.18,0.22,0.27,0.33] \\
     5 & [0.13,0.16,0.195,0.235,0.28] \\
     6 & [0.095,0.115,0.14,0.175,0.215,0.26] \\
     7 & [0.075,0.09,0.11,0.13,0.16,0.195,0.24 \\ 
     8 & [0.055,0.065,0.08,0.1,0.125,0.155,0.19,0.23]
\end{tabular}
\caption{Data for the near optimal catalyst results presented in Fig.\ \ref{fig:catalyst-plots} in the case $p=0.5$ and $q=0.6$. $d$ stands for the dimension of the distribution. Each distribution $r$ was found via numerical search as described in the main text.}
\end{table}

\begin{table}[H]
\centering
\begin{tabular}{r|l}
     $d$ & Lower Bound Embezzling Distribution $r$   \\ \hline
     1 &  n/a \\
     2 & [0.5,0.5] \\
     3 &	 [0.21,0.32,0.47] \\
     4 & [0.12,0.185,0.28,0.415] \\
     5 & [0.075,0.115,0.175,0.26,0.375] \\
     6 & [0.05,0.075,0.11,0.165,0.245,0.355] \\
     7 & [0.03,0.045,0.07,0.105,0.16,0.24,0.35] \\
     8 & [0.05,0.065,0.085,0.095,0.13,0.14,0.19,0.245]
\end{tabular}
\caption{Data for the near optimal catalyst results presented in Fig.\ \ref{fig:catalyst-plots} in the case $p=0.5$ and $q=0.7$. $d$ stands for the dimension of the distribution. Each distribution $r$ was found via numerical search as described in the main text.}
\end{table}

\begin{table}[H]
\centering
\begin{tabular}{r|l}
     $d$ & Lower Bound Embezzling Distribution $r$   \\ \hline
     1 &  n/a \\
     2 & [0.5,0.5] \\
     3 &	 [0.19,0.31,0.5]\\
     4 & [0.1,0.165,0.275,0.46] \\
     5 & [0.055,0.095,0.155,0.26,0.435] \\
     6 & [0.07,0.11,0.12,0.185,0.2,0.315] \\
     7 & [0.04,0.065,0.105,0.115,0.175,0.19,0.31] \\
     8 &	 [0.035,0.06,0.065,0.1,0.11,0.165,0.18,0.285]
\end{tabular}
\caption{Data for the near optimal catalyst results presented in Fig.\ \ref{fig:catalyst-plots} in the case $p=0.6$ and $q=0.65$. $d$ stands for the dimension of the distribution. Each distribution $r$ was found via numerical search as described in the main text.}
\end{table}

\begin{table}[H]
\centering
\begin{tabular}{r|l}
     $d$ & Lower Bound Embezzling Distribution $r$   \\ \hline
     1 &  n/a \\
     2 & [0.5,0.5] \\
     3 & [0.24,0.275,0.485] \\
     4 & [0.08,0.15,0.27,0.5] \\
     5 & [0.04,0.075,0.14,0.26,0.485] \\
     6 & [0.055,0.095,0.11,0.18,0.205,0.355] \\
     7 & [0.03,0.055,0.09,0.11,0.17,0.2,0.345 \\
     8 & [0.03,0.05,0.055,0.09,0.1,0.165,0.185,0.325]
\end{tabular}
\caption{Data for the near optimal catalyst results presented in Fig.\ \ref{fig:catalyst-plots} in the case $p=0.6$ and $q=0.7$. $d$ stands for the dimension of the distribution. Each distribution $r$ was found via numerical search as described in the main text.}
\end{table}

\begin{table}[H]
\centering
\begin{tabular}{r|l}
     $d$ & Lower Bound Embezzling Distribution $r$   \\ \hline
     1 &  n/a \\
     2 & [0.5,0.5] \\
     3 &	 [0.21,0.29,0.5] \\
     4 &	 [0.09,0.175,0.24,0.495] \\
     5 &	 [0.05,0.095,0.16,0.23,0.465] \\
     6 & [0.035,0.065,0.09,0.155,0.215,0.44] \\
     7 & [0.02,0.04,0.065,0.09,0.15,0.21,0.425] \\
     8 & [0.015,0.03,0.04,0.065,0.09,0.145,0.205,0.41]
\end{tabular}
\caption{Data for the near optimal catalyst results presented in Fig.\ \ref{fig:catalyst-plots} in the case $p=0.6$ and $q=0.8$. $d$ stands for the dimension of the distribution. Each distribution $r$ was found via numerical search as described in the main text.}
\end{table}

}

\end{document}